\documentclass{article}

\usepackage[hidelinks,hypertexnames=false,pagebackref]{hyperref}

\renewcommand*{\backref}[1]{}
\renewcommand*{\backrefalt}[4]{$\hookrightarrow$~\ifnum#1=1{p.}\else{pp.}\fi~#2}
\usepackage[numbers,sort]{natbib}
\usepackage{enumitem}
\usepackage{geometry}
\usepackage{mathtools}
\usepackage{amsthm}
\usepackage{amsfonts}
\usepackage{amssymb}

\newcommand{\diff}{\mathop{}\!\mathrm d}
\newcommand{\dx}{\diff x}
\newcommand{\dy}{\diff y}
\newcommand{\R}{\mathbb R}
\newcommand{\N}{\mathbb N}
\newcommand{\Z}{\mathbb Z}
\newcommand{\I}{\mathbf 1}
\newcommand{\Tor}{\mathbb T}
\newcommand{\Holder}{\mathcal C}
\newcommand{\wick}[1]{\mathord{\mathopen: #1 \mathclose:}}
\newcommand{\abs}[1]{|#1|}
\newcommand{\bigabs}[1]{\left|#1\right|}
\newcommand{\norm}[1]{\|#1\|}
\newcommand{\bignorm}[1]{\left\|#1\right\|}
\newcommand{\inorm}[1]{\langle #1 \rangle}
\newcommand{\E}{\mathord{\mathbb E}}
\newcommand{\Prob}{\mathord{\mathbb P}}
\DeclareMathOperator{\Laplace}{\mathrm{\Delta}}
\DeclareMathOperator{\Hess}{Hess}

\newcommand{\ExtF}{\mathcal{\tilde E}}
\newcommand{\proj}{\mathrm P}

\newtheorem{theorem}{Theorem}[section]
\newtheorem{lemma}[theorem]{Lemma}

\newtheorem{proposition}[theorem]{Proposition}
\theoremstyle{definition}
\newtheorem{definition}[theorem]{Definition}
\newtheorem{remark}[theorem]{Remark}
\newtheorem{assumption}{Assumption}
\numberwithin{equation}{section}

\title{$L^2$ convergence rates for singular stochastic wave dynamics}
\author{Nikolay Barashkov%
\footnote{Max Planck Institute for Mathematics in the Sciences, Germany}
\and Giacomo Di Gesù%
\footnote{Department of Mathematics, Sapienza Università di Roma, Italy}
\and Petri Laarne%
\footnote{Department of Mathematics and Statistics, University of Helsinki, Finland}}
\date{}

\begin{document}

\maketitle
\begin{abstract}
\noindent
We establish quantitative hypocoercivity for singular stochastic damped wave equations arising in stochastic quantization.
Building on the finite-dimensional argument of Cao, Lu, and Wang [Arch.~Rat.~Mech.~Anal.~2023],
we derive a dimension-free criterion that replaces uniform semiconvexity by a log-Sobolev inequality together with exponential lower-tail bounds for the bottom of the Hessian spectrum. We verify this criterion for the two-dimensional $\Phi^4_2$ and sine-Gordon models and for the three-dimensional $\Phi^4_3$ model, through estimates on the lowest eigenvalue of associated singular random Schr\"odinger operators. Combined with the continuum limit established in a companion paper, this yields exponential $L^2$ convergence to equilibrium for the $\Phi^4_2$ dynamics and for sine-Gordon with $\beta^2<2\pi$. For $\Phi^4_2$, we verify that, in a weak-coupling, small-mass regime, optimizing the damping yields a kinetic relaxation rate of at least order $m$. This yields a genuine acceleration over the order-$m^2$ relaxation rate of the overdamped dynamics.
The lattice estimates are uniform in the discretization for sine-Gordon throughout $\beta^2<4\pi$, and for $\Phi^4_3$ at fixed volume, including regimes beyond those for which the continuum dynamics is presently constructed.


\end{abstract}

\section{Introduction}
Stochastic quantization is a method for constructing infinite-dimensional probability measures as stationary measures of stochastic evolution equations~\cite{parisi_perturbation_1981}. A central class of examples arises in Euclidean quantum field theory, where one is led to Gibbs measures of the form
\begin{equation}\label{eq:generic measure continuum}
\mu(du)
=
Z^{-1}
\exp\left(
-\int
\left(
V(u(x))
+\frac{m^2}{2}u(x)^2
+\frac12|\nabla u(x)|^2
\right)\,dx
\right)
\prod_x du(x).
\end{equation}
Here $m>0$ is the mass parameter and $V$ is the interaction potential.
The associated stochastic quantization dynamics, with stationary measure $\mu$, is the stochastic heat equation
\begin{equation}\label{eq:sdnlw continuum parabolic}
\partial_t u(x,t)+(m^2-\Delta)u(x,t)
=
-V'(u(x,t))+\sqrt{2}\,\xi(x,t),
\end{equation}
where $\xi$ is space-time white noise. 
Adjoining a velocity variable leads to the kinetic counterpart, given by the stochastic damped wave equation~\cite{ryang_canonical_1985}
\begin{equation}\label{eq:sdnlw continuum}
\partial_{tt}u(x,t)
+\gamma\partial_tu(x,t)
+(m^2-\Delta)u(x,t)
=
-V'(u(x,t))+\sqrt{2\gamma}\,\xi(x,t),
\end{equation}
where $\gamma>0$ is the damping parameter. Its invariant measure on phase space is $\mu\otimes\kappa$, where $\kappa$ is the Gaussian white noise measure in the velocity variable.

Equation~\eqref{eq:sdnlw continuum} is infinite-dimensional kinetic Langevin dynamics,
whereas \eqref{eq:sdnlw continuum parabolic} is its overdamped counterpart: after the natural rescaling of time, the latter arises from the former in the large-damping, or Smoluchowski--Kramers, limit~\cite{cerrai_smoluchowski_2006,zine_smoluchowski_2025}.
In the overdamped case, exponential relaxation in $L^2(\mu)$ is governed by the Poincaré constant $\sigma_p$ of the equilibrium measure, equivalently the spectral gap of the corresponding reversible dynamics. The situation in the kinetic case is more subtle: noise and dissipation act directly only on the velocity variable, while relaxation of the position variable is induced through its coupling with the velocity.
The resulting mechanism of convergence to equilibrium is known as hypocoercivity and is by now a classical subject; see for instance the earlier monographs \cite{helffer_nier_2005,villani_hypocoercivity_2009} and the references therein.
For a more recent overview of the extensive literature on quantitative hypocoercivity for kinetic Langevin dynamics, see e.g.\ the introductions of \cite{cao_explicit_2023}
and~\cite{lu_2026_quantitative}.

Of particular importance for us is the finite-dimensional result of
Cao, Lu, and Wang~\cite{cao_explicit_2023}, which gives $L^2$
convergence rates in terms of the Poincaré constant, the damping parameter, and quantitative control of the Hessian of the potential.
Under a semiconvexity assumption their estimates are dimension-free. Moreover, if the semiconvexity constant is bounded by a constant multiple of $\sigma_p$, as in particular in the convex case, then optimizing the damping yields a kinetic rate of order $\sqrt{\sigma_p}$, compared with the optimal overdamped rate $\sigma_p$. This represents a substantial acceleration in the small-gap regime.

The observed square-root acceleration can be viewed as an instance of a more general diffusive-to-ballistic phenomenon associated with non-reversible lifts of reversible dynamics. In particular, Eberle and L\"orler~\cite{eberle_non_2026} show that for a general class of second-order lifts, the relaxation time can improve at most by a square-root factor, with kinetic Langevin dynamics as a canonical example.

These considerations naturally raise the question whether hypocoercive acceleration persists in infinite dimensions, in particular for singular stochastic quantization models. 
In the singular parabolic setting, exponential relaxation has been established
by several complementary methods, including functional inequalities,
semigroup contraction estimates, and coupling or synchronization
arguments~\cite{bauerschmidt_logsobolev_2021,
bauerschmidt_logsobolev_2024,kunick_gradienttype_2025,
tsatsoulis_spectral_2018,duch_ergodicity_2025}.
The singular hyperbolic setting is considerably less developed. Existing results for stochastic wave equations range from unique ergodicity to exponential mixing in dual-Lipschitz or Wasserstein-type distances~\cite{forlano_2024_on,liu_2026_spectral,martirosyan_2014_exponential,seong_2024_exponential,tolomeo_2020_unique}. To our knowledge, however, quantitative $L^2$ hypocoercive estimates for singular stochastic quantization dynamics that allow a direct comparison with the overdamped dynamics have not previously been available. The main contribution of this paper is to obtain such a comparison in the singular hyperbolic setting.

Our approach is to start from the finite-dimensional argument of Cao, Lu, and Wang and derive a dimension-independent criterion for quantitative $L^2$ hypocoercivity in which semiconvexity is replaced by a log-Sobolev inequality and exponential lower-tail control of the bottom of the Hessian spectrum. This replacement is crucial for the singular models considered here, for which uniform semiconvexity is lost along the lattice approximations. The verification of the new criterion leads naturally to spectral estimates for singular random Schrödinger operators.

\subsection{Setting and main results}

We consider~\eqref{eq:sdnlw continuum} on the $d$-dimensional torus
$\Tor_L^d$ of side length $L\geq 1$, for $d=2,3$.
In two dimensions we study two models:
\begin{itemize}
\item the $\Phi^4_2$ model, corresponding to the choice
$V(u)=\lambda\wick{u^4}$ for $\lambda>0$, and
\item the sine-Gordon model, corresponding to the choice
$V(u)=\lambda\wick{\cos(\beta u)}$ for $\lambda>0$.
\end{itemize}
Here $\wick{\cdot}$ denotes Wick ordering with respect to the Gaussian
field with covariance $(m^2-\Laplace)^{-1}$;
see Appendix~\ref{sec:wick} for the relevant definitions and identities.
Throughout the paper, the sine-Gordon parameter $\beta$ is fixed and satisfies
\begin{equation}
0<\beta^2<4\pi.
\end{equation}
In three dimensions we consider the $\Phi^4_3$ model,
corresponding to the choice $V(u)=\lambda\wick{u^4}$ for $\lambda>0$,
but requiring an additional renormalization: at lattice scale $\delta$,
the interaction takes the form
$V_\delta(u)=\lambda\wick{u^4}+\lambda^2 a_\delta\wick{u^2}$,
where $a_\delta$ diverges logarithmically as $\delta\to0$.

For each of the three models, we denote by $\mu$ the corresponding
continuum Gibbs measure, whose existence is well established.

Global stationary solutions
$(u(t),\partial_t u(t))\in
C([0,T],H^{-\varepsilon}(\Tor_L^2)\times H^{-1-\varepsilon}(\Tor_L^2))$,
for every $T>0$ and $\varepsilon>0$, were constructed for the
$\Phi^4_2$ model in~\cite{gubinelli_global_2022}, and for the
sine-Gordon model with $\beta^2<2\pi$ in~\cite{oh_invariant_2021}.
For sine-Gordon, the construction of stationary solutions has
subsequently been extended slightly beyond $\beta^2=2\pi$
in~\cite{zine_hyperbolic_2025}.
Their fixed-time law is $\mu\otimes\kappa$, where $\kappa$ is
spatial white noise.
At fixed times, the corresponding Sobolev regularity fails
almost surely at $\varepsilon=0$; in particular, the position
field is distribution-valued, and the nonlinearities require
renormalization.

For the $\Phi^4_3$ model, global solutions preserving the invariant
measure have been constructed for~\eqref{eq:sdnlw continuum}
with $\gamma=0$ in~\cite{bringmann_invariant_2024}.
For $\gamma>0$, the construction of the corresponding continuum
dynamics with space-time white noise remains open.
Our results in three dimensions concern lattice approximations.

To state our main result, we first recall the log-Sobolev
inequalities established for lattice approximations of the
equilibrium measures $\mu$ of these three models
in~\cite{bauerschmidt_logsobolev_2021,bauerschmidt_logsobolev_2024}.
For $\delta>0$ with $L/\delta\in\N$, let
$\Lambda_{\delta,L}=\Tor_L^d\cap\delta\Z^d$ be the periodic lattice,
and denote by $\mu_\delta$ the corresponding lattice approximation
of $\mu$, see~\eqref{eq:generic measure} below.
We write $\Lambda_\delta=\Lambda_{\delta,L}$ when $L$ is understood.
For $\sigma>0$, the measure $\mu_\delta$ satisfies a log-Sobolev
inequality with constant $\delta^d\sigma$ if
\begin{equation}\label{eq:lsi-lattice}
\int f^2\log(f^2)\diff\mu_\delta
\leq \frac{2}{\delta^d\sigma}
\int\abs{\nabla f}^2\diff\mu_\delta
\end{equation}
for every $f\in C_b^\infty(\R^{\Lambda_\delta})$ with
$\int f^2\diff\mu_\delta=1$.
For the $\Phi^4_2$ and $\Phi^4_3$ measures, Bauerschmidt and
Dagallier~\cite{bauerschmidt_logsobolev_2024} prove that the
log-Sobolev inequality~\eqref{eq:lsi-lattice} holds with
$\sigma>0$ independent of $\delta$ for arbitrary coupling at
fixed finite volume, and with $\sigma$ independent also of $L$
throughout the high-temperature phase, characterized by
susceptibility bounded uniformly in the lattice regularization
and volume.
More quantitatively, in the weak-coupling regime
$\lambda\leq\lambda_0(m)$, which is contained in this phase,
their estimates give a $\sigma$ bounded below by
$c^{-1}m^2$ for some universal constant $c>0$, uniformly in $L$.

For the sine-Gordon measure, Bauerschmidt and
Bodineau~\cite{bauerschmidt_logsobolev_2021} prove, for
$\beta^2<6\pi$ and fixed $m,L,\lambda$, a log-Sobolev inequality
with a positive $\sigma$ uniform in the lattice regularization.
Moreover, if $Lm\geq1$ and
$\lambda\leq c_\beta m^{2-\beta^2/(4\pi)}$ for $c_\beta>0$
sufficiently small, their estimate gives a $\sigma$
bounded below by a constant multiple of $m^2$, uniformly in $L$.

We shall also consider the lattice Poincar\'e inequality
\begin{equation}\label{eq:poincare-lattice}
\int\abs{\nabla f}^2\diff\mu_\delta
\geq \delta^d\sigma_p\int f^2\diff\mu_\delta
\end{equation}
for every $f\in C_b^\infty(\R^{\Lambda_\delta})$ with
$\int f\diff\mu_\delta=0$, where $\sigma_p>0$. The choice $\sigma_p=\sigma$ is always
admissible whenever $\sigma$ satisfies~\eqref{eq:lsi-lattice}.

Our main result concerns the rate of convergence to equilibrium
in $L^2(\mu\otimes\kappa)$ of the continuum dynamics~\eqref{eq:sdnlw continuum}
for the $\Phi^4_2$ model and the
sine-Gordon model with $\beta^2<2\pi$.
We denote by $(P_t)_{t\geq 0}$ the associated Markov semigroup
on $L^2(\mu\otimes\kappa)$, whose action is given,
for $(\mu\otimes\kappa)$-almost every $(u_0,v_0)$, by
\begin{equation}
P_t f(u_0,v_0)
= \E\bigl[f(u(t),\partial_t u(t)) \mid
(u(0),\partial_t u(0))=(u_0,v_0)\bigr].
\end{equation}

\begin{theorem}[Equilibration of dynamics]\label{thm:main phi4}
For every $L\geq1$ and $m,\lambda>0$, consider the
$\Phi^4_2$ model or the sine-Gordon model with fixed
$\beta^2<2\pi$ on $\Tor_L^2$.
Fix $\sigma_p\geq\sigma>0$ such that
\eqref{eq:lsi-lattice} and~\eqref{eq:poincare-lattice}
hold uniformly in $\delta$.

There exist finite constants $C>0$ and
$R=R(m,\lambda,\sigma)$, with $C$ universal and $R$
possibly depending on $L$ only through $\sigma$, such that
for every $\gamma>0$, $t\geq0$, and
$f\in L^2(\mu\otimes\kappa)$ with
$\|f\|_{L^2(\mu\otimes\kappa)}=1$,
\begin{equation}\label{eq:main-convergence}
\left\|P_t f-\int f\diff(\mu\otimes\kappa)\right\|_{L^2(\mu\otimes\kappa)}
\leq C\exp\left(
-\frac{\sigma_p\gamma t}
{C(\sqrt{\sigma_p}+R+\gamma)^2}
\right).
\end{equation}
Moreover, for the $\Phi^4_2$ model and every
$\varepsilon>0$, there exists a constant
$C_{L,\varepsilon}>0$ such that the estimate
\eqref{eq:main-convergence} also holds with
$R\leq C_{L,\varepsilon}m$, for all
$0<m\leq1$ and
$0<\lambda\leq\lambda_0(m)\wedge m^{4+\varepsilon}$.
\end{theorem}

\begin{remark}\label{remark:comparison}
The lattice log-Sobolev and Poincar\'e inequalities pass to the
continuum limit, so that $\sigma$ and $\sigma_p$ are also
admissible constants for the corresponding inequalities under
$\mu$. Denoting by $\sigma_p^{\mathrm{opt}}(\mu)$ the optimal
Poincar\'e constant of $\mu$, we thus have
$\sigma_p^{\mathrm{opt}}(\mu)\geq\sigma_p\geq\sigma$.
For the $\Phi^4_2$ model in the weak-coupling regime of
Theorem~\ref{thm:main phi4}, the susceptibility estimates of
Bauerschmidt and Dagallier~\cite{bauerschmidt_logsobolev_2024}
give $\sigma_p^{\mathrm{opt}}(\mu)\lesssim m^2$,
uniformly in $L$, by testing the Poincar\'e inequality with
the spatial average of the field. Since one may choose
$\sigma_p=\sigma\gtrsim m^2$, it follows that
$\sigma_p^{\mathrm{opt}}(\mu) \simeq m^2$.
Together with the bound $R\leq C_{L,\varepsilon}m$,
Theorem~\ref{thm:main phi4} then yields a kinetic relaxation
rate of at least
\begin{equation}
\frac{m^2\gamma}{C_{L,\varepsilon}(m+\gamma)^2}.
\end{equation}
Choosing $\gamma=m$ yields a kinetic relaxation rate of at least
order $m$, compared with the optimal overdamped relaxation rate
$\sigma_p^{\mathrm{opt}}(\mu) \simeq m^2$. 
This establishes the announced acceleration in the small-mass regime.
\end{remark}

\begin{remark}\label{rem:regular-one-dimensional}
In one spatial dimension, analogous $L^2$ convergence estimates
follow from the dimension-free result of Cao, Lu, and
Wang~\cite{cao_explicit_2023}. For a smooth, bounded-below
interaction potential $V$ satisfying suitable growth conditions
and $\Hess V\geq-R^2$, their result yields, provided the
corresponding lattice Gibbs measures satisfy a uniform
Poincar\'e inequality with constant $\sigma_p>0$,
a relaxation rate
\[
\alpha\geq
\frac{\sigma_p\gamma}
{C(\sqrt{\sigma_p}+R+\gamma)^2},
\]
uniformly in the lattice discretization.
The corresponding continuum estimate follows by passing
to the limit as in~\cite{LATTICE}.
\end{remark}

 \bigskip\noindent
 We now formulate a general, dimension-free hypocoercivity
criterion for lattice kinetic Langevin dynamics, which will
serve as the basis for the proof of Theorem~\ref{thm:main phi4}.
We work on the periodic lattice $\Lambda_\delta$ defined above,
allowing more generally a lattice-dependent interaction
potential $V_\delta:\R\to\R$.
We assume that $V_\delta$ is bounded below for each $\delta$,
without requiring a uniform lower bound.
The corresponding lattice Gibbs measure $\mu_\delta$ and
Gaussian velocity measure $\kappa_\delta$ are given by
\begin{equation}\label{eq:generic measure}
\begin{aligned}
\mu_\delta(du)
&\propto \exp(-U_\delta(u))
\prod_{x\in\Lambda_\delta}du(x),\\
\text{where}\quad U_\delta(u)
&=\delta^d\sum_{x\in\Lambda_\delta}
\left[
V_\delta(u(x))
+\frac{m^2}{2}u(x)^2
+\frac12\abs{\nabla_\delta u(x)}^2
\right],\\
\kappa_\delta(dv)
&\propto
\exp\left(
-\frac12\sum_{x\in\Lambda_\delta}v(x)^2
\right)
\prod_{x\in\Lambda_\delta}dv(x).
\end{aligned}
\end{equation}
For brevity, we write $U=U_\delta$ whenever $\delta$ is understood.
The product measure $\mu_\delta\otimes\kappa_\delta$
is invariant under the lattice kinetic Langevin dynamics
\begin{equation}\label{eq:sde lattice}
\partial_{t't'}^2 u_\delta
+\tilde\gamma\,\partial_{t'}u_\delta
+\delta^d(m^2-\Laplace_\delta)u_\delta
=-\delta^d V_\delta'(u_\delta)
+\sqrt{2\tilde\gamma}\,\xi_\delta,
\end{equation}
where $\tilde\gamma>0$ and $t'$ is an auxiliary time variable,
to be rescaled below to match the continuum dynamics.
Here $\Laplace_\delta$ is the periodic lattice Laplacian on
$\Lambda_\delta$, and
$(\xi_\delta(x,\cdot))_{x\in\Lambda_\delta}$ are independent
white noises in time, normalized by
$\E[(\int_0^t\xi_\delta(x,s)\diff s)^2]=t$ for $t\geq0$.
We denote by $(P_t^\delta)_{t\geq0}$ the corresponding
Markov semigroup on $L^2(\mu_\delta\otimes\kappa_\delta)$.
Our first assumption concerns the functional inequalities
recalled above.

\begin{assumption}\label{ass:log-sobolev}
There exist constants $\sigma_p\geq\sigma>0$, independent
of $\delta$, such that the measures $\mu_\delta$ satisfy
the log-Sobolev inequality~\eqref{eq:lsi-lattice} and the
Poincar\'e inequality~\eqref{eq:poincare-lattice}.
\end{assumption}

We next introduce two alternative conditions controlling
the lower spectral tails of the Hessian of $U$,
one global and the other localized.
Here and below, $\Hess U$ denotes the Hessian with respect
to the lattice $L^2(\Lambda_\delta)$ inner product, so that
\begin{equation}\label{eq:hessian}
\Hess U(u)=V_\delta''(u)+m^2-\Laplace_\delta.
\end{equation}

\renewcommand{\theassumption}{\arabic{assumption}A}
\begin{assumption}[Simple version]\label{ass:hessian finite}
For every $\delta$, the interaction potential
$V_\delta$ belongs to $C^2(\R)$.
There exists a finite constant $R\geq0$, independent of
$\delta$, such that
\[
\log\int
\exp\left(\frac{4}{\sigma}Q_\delta(u)\right)
\diff\mu_\delta(u)
\leq\frac{4R^2}{\sigma},
\]
where $Q_\delta(u)=-\lambda_{\min}(\Hess U(u))$ is the
negative of the lowest eigenvalue of $\Hess U(u)$,
and $\sigma$ is as in Assumption~\ref{ass:log-sobolev}.
\end{assumption}
Assumption~\ref{ass:hessian finite} is weaker than the
uniform semiconvexity condition used in~\cite{cao_explicit_2023}.
For the models considered here, however, one cannot generally
choose $R$ uniformly in $L$.
To obtain volume-uniform estimates, we introduce the following
localized alternative.

\addtocounter{assumption}{-1}
\renewcommand{\theassumption}{\arabic{assumption}B}
\begin{assumption}[Uniform version]\label{ass:potential-intro}
For every $\delta$, the interaction potential
$V_\delta$ belongs to $C^2(\R)$.
There exist functions
$Q_{\delta,i}:\R^{\Lambda_\delta}\to\R$,
indexed by $i\in\Z^d$, and a finite constant $R\geq0$,
independent of $\delta$ and $L$ except through $\sigma$,
such that for every $i\in\Z^d$, $u\in\R^{\Lambda_\delta}$,
and $h\in\R^{\Lambda_\delta}$ supported in $B(i,2)$, we have
\begin{equation}\label{eq:assumption on hessian}
\delta^d\sum_{x\in\Lambda_\delta}
\left[
\bigl(V_\delta''(u(x))+m^2\bigr)h(x)^2
+\frac12\abs{\nabla_\delta h(x)}^2
\right]
\geq -Q_{\delta,i}(u)\norm{h}_{L^2(\Lambda_\delta)}^2,
\end{equation}
and
\begin{equation}\label{eq:potential assumption exponential}
\sup_{i\in\Z^d}
\log\int
\exp\left(\frac{4}{\sigma}Q_{\delta,i}(u)\right)
\diff\mu_\delta(u)
\leq\frac{4R^2}{\sigma},
\end{equation}
where $\sigma$ is as in Assumption~\ref{ass:log-sobolev}.
\end{assumption}
\renewcommand{\theassumption}{\arabic{assumption}}

Like Assumption~\ref{ass:hessian finite},
Assumption~\ref{ass:potential-intro} does not require uniform
semiconvexity. Unlike the global condition, however,
\eqref{eq:assumption on hessian} retains only half of
the gradient term in~\eqref{eq:hessian}.
This permits localization and yields estimates whose dependence
on the volume $L$ enters only through $\sigma$.
The price is an additional term in Lemma~\ref{thm:clw 2.3}:
the resulting localization error contributes an order-one
term to the effective Hessian parameter in the hypocoercivity
estimate. Consequently, this argument alone does not yield
the $O(m)$ bound needed for small-mass acceleration uniformly
in $L$. The acceleration result discussed in
Remark~\ref{remark:comparison} is therefore established
here only for fixed finite volume.

Our general criterion yields the following dimension-free
$L^2$ convergence estimate.

\begin{theorem}[Lattice hypocoercivity]
\label{thm:decay of variance}
Under Assumption~\ref{ass:log-sobolev} and either
Assumption~\ref{ass:hessian finite} or~\ref{ass:potential-intro}, together with
the technical Assumption~\ref{ass:technical} stated
in Section~\ref{sec:clw}, there exists a universal
finite constant $C>0$ such that for every
$\delta>0$ with $L/\delta\in\N$,
$\tilde\gamma>0$, $t'\geq0$, and
$f\in L^2(\mu_\delta\otimes\kappa_\delta)$,
\[
\left\|P_{t'}^\delta f-\int f\diff(\mu_\delta\otimes\kappa_\delta)\right\|_{L^2(\mu_\delta\otimes\kappa_\delta)}
\leq C e^{-\alpha_\delta t'}
\left\|f-\int f\diff(\mu_\delta\otimes\kappa_\delta)\right\|_{L^2(\mu_\delta\otimes\kappa_\delta)},
\]
where
\[
\alpha_\delta=
\frac{\sigma_p\delta^d\tilde\gamma}
{C\bigl[\delta^{d/2}(\sqrt{\sigma_p}+\tilde R)
+\tilde\gamma\bigr]^2},
\]
with $\tilde R=R$ under Assumption~\ref{ass:hessian finite}
and $\tilde R=\sqrt{R^2+C_d}$ under
Assumption~\ref{ass:potential-intro},
where $C_d$ only depends on the dimension of the lattice~$\Lambda$.
\end{theorem}
 
 We next establish the spectral estimates needed to apply
Theorem~\ref{thm:decay of variance} to the singular models
considered here. These estimates concern the lowest
eigenvalues of associated random Schr\"odinger operators;
see Remark~\ref{rem:bottom-spectrum} below.
In two dimensions, we obtain localized
estimates uniform in the volume, verifying
Assumption~\ref{ass:potential-intro}.
For the $\Phi^4_2$ model, we also establish a global estimate
at fixed volume in a weak-coupling regime, which verifies
Assumption~\ref{ass:hessian finite} with the $O(m)$ bound
needed for small-mass acceleration.

\begin{theorem}[Spectral estimates in two dimensions]
\label{thm:integrability}
Let $d=2$, and consider the lattice $\Phi^4_2$ model or
the lattice sine-Gordon model with fixed $0<\beta^2<4\pi$.

\begin{enumerate}
\item[(i)] For every $L\geq1$, $m,\lambda>0$, and
$\delta>0$ with $L/\delta\in\N$, there exist functions
$Q_{\delta,i}:\R^{\Lambda_\delta}\to\R$,
indexed by $i\in\Z^2$, such that for every
$i\in\Z^2$, $u\in\R^{\Lambda_\delta}$,
$h\in\R^{\Lambda_\delta}$ supported in $B(i,2)$,
and $q>0$,
\begin{equation}\label{eq:integrability bounds}
\begin{split}
\delta^2\sum_{x\in\Lambda_\delta}
\left[
\bigl(V_\delta''(u(x))+m^2\bigr)h(x)^2
+\frac12\abs{\nabla_\delta h(x)}^2
\right]
&\geq
-Q_{\delta,i}(u)\norm{h}_{L^2(\Lambda_\delta)}^2,\\
\sup_{i\in\Z^2}\log\int
\exp\bigl(qQ_{\delta,i}(u)\bigr)\diff\mu_\delta(u)
&\leq \mathcal R(q,m,\lambda),
\end{split}
\end{equation}
where $\mathcal R(q,m,\lambda)<\infty$ is independent
of $\delta$ and $L$.

\item[(ii)] For the lattice $\Phi^4_2$ model, let
$Q_\delta(u)=-\lambda_{\min}(\Hess U(u))$,
as in Assumption~\ref{ass:hessian finite}.
There exists a universal constant $c_0>0$ such that
for every $L\geq1$ and $\varepsilon>0$, there is a
constant $C_{L,\varepsilon}<\infty$ with the following
property: for all $0<m\leq1$ and
$0<\lambda\leq\lambda_0(m)\wedge m^{4+\varepsilon}$,
\begin{equation}\label{eq:global spectral estimate phi4}
\sup_{\substack{\delta>0\\L/\delta\in\N}}
\log\int
\exp\left(\frac{4}{c_0m^2}Q_\delta(u)\right)
\diff\mu_\delta(u)
\leq C_{L,\varepsilon}.
\end{equation}
\end{enumerate}
\end{theorem}

\begin{proof}
Item~(i) is proved in Section~\ref{sec:measures}.
Item~(ii) follows immediately, since the additional $\tfrac 1 2 \abs{\nabla_\delta h}^2$ term
only improves the estimates.
\end{proof}

\begin{remark}\label{rem:global-spectral-estimate}
The global estimate in Theorem~\ref{thm:integrability}(ii)
verifies Assumption~\ref{ass:hessian finite} with
$R=O(m)$ at fixed volume.
Indeed, in the weak-coupling regime under consideration,
$\sigma_*=c_0m^2$ is an admissible log-Sobolev constant,
uniformly in $\delta$.
By~\eqref{eq:global spectral estimate phi4},
Assumption~\ref{ass:hessian finite} therefore holds with
\[
R^2=\frac{\sigma_*}{4}C_{L,\varepsilon}
=\frac{c_0C_{L,\varepsilon}}{4}m^2.
\]
This is the bound needed to establish the small-mass
acceleration discussed in Remark~\ref{remark:comparison}.
\end{remark}
Theorem~\ref{thm:integrability}(i) also covers the
sine-Gordon model throughout $2\pi<\beta^2<4\pi$.
Combined with the lattice log-Sobolev inequality
of~\cite{bauerschmidt_logsobolev_2021}
and
Theorem~\ref{thm:decay of variance}, this yields
$L^2$ convergence estimates uniform in the lattice
discretization, including values of $\beta$ for which
the continuum dynamics has not yet been constructed.

For the $\Phi^4_3$ model, we likewise obtain
$L^2$ convergence estimates uniform in the lattice
discretization, but restrict here to fixed volume $L=1$.
The key ingredient is a global spectral estimate for
the Hessian of the lattice potential.
The additional renormalization counterterm $a_\delta$
is specified e.g.\ in~\cite{gubinelli_pde_2021}.

\begin{theorem}[Spectral estimate in three dimensions]
\label{thm:3 dimensions}
Let $d=3$ and $L=1$. For every $m,\lambda>0$ and
$\delta>0$ with $1/\delta\in\N$, consider the
lattice $\Phi^4_3$ model with interaction
$
V_\delta(u)=\lambda\wick{u^4}
+\lambda^2a_\delta\wick{u^2}.
$
Let $Q_\delta(u)=-\lambda_{\min}(\Hess U(u))$, so that
\begin{equation}\label{eq:global hessian phi3}
Q_\delta(u)=
-\inf_{\|\psi\|_{L^2(\Lambda_\delta)}=1}
\delta^3\sum_{x\in\Lambda_\delta}
\left[
\bigl(V_\delta''(u(x))+m^2\bigr)\psi(x)^2
+\abs{\nabla_\delta\psi(x)}^2
\right].
\end{equation}
For every $q>0$,
\[
\log\int\exp\bigl(qQ_\delta(u)\bigr)
\diff\mu_\delta(u)
\leq \mathcal R(q,m,\lambda),
\]
where $\mathcal R(q,m,\lambda)<\infty$ is independent
of $\delta$.
Moreover,
\[
\sup_{\substack{
0<m\leq1,\ 0<\lambda\leq m^{3}\\
0<q\leq m^{-2}
}}
\mathcal R(q,m,\lambda)<\infty.
\]
\end{theorem}

\begin{proof}
Section~\ref{sec:phi43}.
\end{proof}

\begin{remark}\label{rem:bottom-spectrum}
The functions $Q_{\delta,i}$ in
Theorem~\ref{thm:integrability}(i) can be chosen as
negatives of the lowest eigenvalues of modified discrete
Schr\"odinger operators with localized potentials.
By contrast, $-Q_\delta$ in
Theorems~\ref{thm:integrability}(ii)
and~\ref{thm:3 dimensions} is the lowest eigenvalue
of $\Hess U$.
These results thus provide exponential lower-tail
estimates for the spectra of singular random
Schr\"odinger operators.

In the $\Phi^4_3$ model, the potential $V_\delta''(u)$
is non-Gaussian and renormalized, with spatial regularity
analogous to that of two-dimensional white noise.
Its spectral analysis is therefore closely related
to that of the Anderson Hamiltonian; see
\cite{allez_continuous_2015,gubinelli_semilinear_2020,
labbe_continuous_2019}.
For the white-noise Anderson Hamiltonian, exponential
lower-tail estimates for the lowest eigenvalue are known;
see~\cite{allez_continuous_2015,labbe_continuous_2019}.
To our knowledge, such estimates have not previously been established for the non-Gaussian, renormalized potential considered here. They may be of independent interest.
\end{remark}

\begin{remark}
To obtain convergence estimates uniform also in the volume
for the $\Phi^4_3$ model, one would need to use the
localized criterion of
Assumption~\ref{ass:potential-intro}.
This requires retaining only half of the gradient term,
as in~\eqref{eq:assumption on hessian}.
However, the renormalization counterterm in $V_\delta''$
is tuned to the full Laplacian in $\Hess U$.
Reducing its coefficient changes the required
renormalization, so the cancellation of divergences
is lost.
It would therefore be interesting to develop a refinement
of Theorem~\ref{thm:decay of variance} that yields
volume-uniform estimates without modifying the
renormalized Hessian.
\end{remark}

To prove Theorem~\ref{thm:main phi4}, it remains to pass
to the continuum limit. We first rescale time and damping
in~\eqref{eq:sde lattice} by setting
$\tilde\gamma=\delta^{d/2}\gamma$ and
$t'=\delta^{-d/2}t$.
This leads to the lattice equation
\begin{equation}\label{eq:sdnlw lattice}
\partial_{tt}u_\delta(x,t)
+\gamma\partial_tu_\delta(x,t)
+(m^2-\Laplace_\delta)u_\delta(x,t)
=-V_\delta'(u_\delta(x,t))
+\sqrt{2\gamma}\,\eta_\delta(x,t),
\end{equation}
where $\eta_\delta$ denotes lattice space-time white noise,
normalized so that its components at distinct sites are
independent and
\begin{equation}
\E\left[
\left(\int_0^t\eta_\delta(x,s)\diff s\right)^2
\right]=\delta^{-d}t.
\end{equation}
By Theorem~\ref{thm:decay of variance}, the corresponding
lattice dynamics converges to equilibrium in $L^2$
with a relaxation rate bounded below by
\begin{equation}
\alpha=
\frac{\sigma_p\gamma}
{C(\sqrt{\sigma_p}+\tilde R+\gamma)^2},
\end{equation}
uniformly in $\delta$.
To pass to the continuum limit in two dimensions,
we use the following result, established in the
companion paper~\cite{LATTICE}.
We denote by $\ExtF$ the Fourier extension operator.

\begin{theorem}[Lattice approximation]
\label{thm:lattice-convergence}
Let $d=2$ and $L\geq1$, and consider either the
$\Phi^4_2$ model or the sine-Gordon model with
$0<\beta^2<2\pi$.
Let $(u,\partial_tu)$ and
$(u_\delta,\partial_tu_\delta)$ be the corresponding
stationary continuum and lattice solutions of
\eqref{eq:sdnlw continuum} and
\eqref{eq:sdnlw lattice}, constructed on a common
probability space.
For every $T>0$ and $\varepsilon>0$, the following
hold almost surely as $\delta\to0$:
\begin{itemize}
\item
$\ExtF u_\delta\to u$ in
$C([0,T],H^{-\varepsilon}(\Tor_L^2))$;
\item
$\ExtF\partial_tu_\delta\to\partial_tu$ in
$C([0,T],H^{-1-\varepsilon}(\Tor_L^2))$.
\end{itemize}
\end{theorem}

Theorem~\ref{thm:lattice-convergence} implies convergence
of the corresponding two-time correlations for bounded
continuous observables on phase space.
We may therefore pass to the limit in the uniform lattice
$L^2$ convergence estimate in bilinear form.
By duality and density in $L^2(\mu\otimes\kappa)$,
this yields Theorem~\ref{thm:main phi4}.

\subsection{Discussion}
Quantitative $L^2$ hypocoercivity is already available
in certain regular infinite-dimensional settings.
In particular, Eisenhuth and
Grothaus~\cite{eisenhuth_grothaus_2024} establish
exponential $L^2$ convergence with explicitly computable
rates for classes of infinite-dimensional Langevin equations,
using a Hilbert-space approach based on the Kolmogorov
generator. This complements the one-dimensional observation
in Remark~\ref{rem:regular-one-dimensional}, where analogous
estimates follow by applying the dimension-free result of
Cao, Lu, and Wang~\cite{cao_explicit_2023} and passing to
the continuum limit.
The main difficulty addressed in the present paper is instead
the singularity of the interactions.

Our method appears to be comparatively insensitive to
ultraviolet divergences, as it yields uniform estimates
even in regimes where the limiting hyperbolic dynamics
has not yet been constructed. Indeed, the hypocoercive
estimates are obtained from functional inequalities
and spectral bounds for the lattice Hamiltonians,
rather than from a pathwise analysis of the limiting
dynamics. This point of view is reminiscent of
generator-based approaches to singular stochastic
equations~\cite{graefner_2023_energy,graefner_2024_weak,
graefner_2025_critical,gubinelli_2018_energy,
gubinelli_2020_the}, and it would be interesting to
investigate to what extent such methods can be developed
for stochastic wave equations.

In parameter regimes where our estimates are also uniform
in the volume, they can be passed to the infinite-volume limit,
provided a suitable convergence of the stationary dynamics
is established; see, for instance,~\cite{oh_2026_hyperbolic}
for the construction of hyperbolic dynamics on the plane.
Such an $L^2$ theory, however, only yields convergence
for initial laws with $L^2$ density with respect to the
equilibrium Gibbs measure. In infinite volume, this is
a severe restriction, since genuinely different spatially
extended states are typically mutually singular.
 A natural further question is
therefore to understand convergence from more general
initial laws. In the parabolic setting, Duch, Hairer,
Yi, and Zhao~\cite{duch_ergodicity_2025} recently
established exponential synchronization for the
infinite-volume $\Phi^4_3$ dynamics at high temperature,
implying in particular uniqueness of the invariant measure.
It would be interesting to investigate whether the
Holley--Stroock uniqueness method, recently implemented
in the parabolic setting in~\cite{bauerschmidt2025holley},
can be adapted to this hyperbolic problem.
In view of the recent hypocoercive entropy-decay estimates
of~\cite{lu_2026_sharp}, this direction appears
particularly promising in the one-dimensional convex case.

\bigskip\noindent
\textbf{Acknowledgements.}
PL was financially supported by the Finnish Centre of Excellence in Randomness and Structures (FiRST).

\bigskip\noindent
\textbf{AI usage disclosure.} ChatGPT 5.6 Sol was used for checking some arguments.
The initial draft and the final revision of the manuscript were fully done by the authors.

\section{Generic estimate on the convergence rate}\label{sec:clw}
The purpose of this section is to establish Theorem~\ref{thm:decay of variance}.
Throughout this section, besides Assumptions~\ref{ass:log-sobolev}
and~\ref{ass:hessian finite}/\ref{ass:potential-intro}
we will make the following qualitative assumption on the potential for finite discretizations,
as in~\cite{cao_explicit_2023}. 

\begin{assumption}\label{ass:technical}
We require that both:

\begin{enumerate}
\item
    $\Laplace_u U(u) \leq M + \dfrac \alpha 2 \abs{\nabla_u U(u)}^2$
    for some $M \geq 1$ and $\alpha \in (0, 1)$, both \emph{possibly dependent} on $\delta$.
\item
    The embeddings $H^1(\mu_\delta) \hookrightarrow L^2(\mu_\delta)$ are compact for all $\delta > 0$.
\end{enumerate}
\end{assumption}

The log-Sobolev inequality implies the following \emph{NGS bound};
see \cite[Theorem~7]{gross_logarithmic_1975} for proof.
We present the result in a rearranged form that matches its usage in Lemma~\ref{thm:clw 2.3}.

\begin{lemma}[NGS~bound]\label{thm:ngs bound}
Assume that the measure $\diff\mu$ satisfies a log-Sobolev inequality with constant $\sigma$.
Let $\psi \colon \R^\Lambda \to \R$ be a continuous function such that
\[
\int \exp\left( \frac{2}{\sigma} \psi(u) \right) \diff\mu(u) \leq K
\]
for some $K > 0$.
Then for any $f \in C^\infty_b(\R^\Lambda)$ we have the estimate
\[
\delta^d \int \psi(u) \abs{f(u)}^2 \diff\mu(u)
\leq \frac{\sigma \delta^d}{2} \log(K) \int \abs{f(u)}^2 \diff\mu(u)
    + \int \abs{\nabla f(u)}^2 \diff\mu(u).
\]
\end{lemma}

Below the time parameter $T$ is as in~\cite{cao_explicit_2023}.
We denote by $\diff\tau$ the normalized Lebesgue measure over the interval $[0,T]$.

\begin{lemma}[{{Analogue of \cite[Lemma~2.1]{cao_explicit_2023}}}]\label{thm:clw 2.1}
Since $\sigma_{p} \geq \sigma$, Assumption~\ref{ass:log-sobolev}
implies that we have the following tensorised Poincaré inequality
for $g\in H^1(\tau \otimes \mu_\delta)$:
\[
\norm{g - {\textstyle\int g \diff(\tau \otimes \mu)}}_{L^2(\tau \otimes \mu)}^2
\leq \max\!\left\{ \frac{1}{\delta^{d}\sigma_{p}}, \frac{T^2}{\pi^2} \right\}
    \left( \norm{\partial_t g}_{L^2(\tau \otimes \mu)}^2
        + \norm{\nabla_u g}_{L^2(\tau \otimes \mu)}^2 \right).
\]
\end{lemma}
\begin{proof}
See e.g.~\cite[Chapter~4]{bakry_analysis_2014}.
\end{proof}

\begin{lemma}[{{Analogue of \cite[Lemma~2.3]{cao_explicit_2023}}}]\label{thm:clw 2.3}
Let $g \in H^2(\tau \otimes \mu_\delta)$ be such that
$(\partial_t g, \nabla_u g) \in H^1_0(\tau \otimes \mu_\delta)^{1 + \abs\Lambda}$.
Then
\[
\begin{gathered}
\norm{D^2 g}_{L^2(\tau \otimes \mu)}^2
\leq 2\norm{(-\partial_{tt} + \nabla_u^\ast \nabla_u) g}_{L^2(\tau \otimes \mu)}^2
    + 2(R^2+2^d) \delta^d \norm{\nabla_u g}_{L^2(\tau \otimes \mu)}^2, \text{ and}\\
\norm{\nabla_u^2 g}_{L^2(\tau \otimes \mu)}^2
\leq 2\norm{\nabla_u^\ast \nabla_u g}_{L^2(\tau \otimes \mu)}^2
    + 2(R^2+2^d) \delta^d \norm{\nabla_u g}_{L^2(\tau \otimes \mu)}^2.
\end{gathered}
\]
Here $D^2$ is the Hessian with respect to all of
$t$ and $u(x)$, $x \in \Lambda$.
If  Assumption~\ref{ass:hessian finite} holds,
then the above bounds hold with $R^2+2^d$ replaced by $R^2$.
\end{lemma}
\begin{proof}
We only prove the second statement, the first is proved analogously.
As in \cite[Eq.~(20)]{cao_explicit_2023}, we start from the following Bochner identity,
also known as the $\Gamma_2$ identity \cite[Section~1.16]{bakry_analysis_2014}:
\begin{equation}\label{eq:clw 2.3 beginning}
\begin{split}
\norm{\nabla_u^2 g}_{L^2(\tau \otimes \mu)}^2
&= \sum_{x,y \in \Lambda} \iint \abs{\partial_{u(x)} \partial_{u(y)} g(u)}^2 \diff\mu_\delta(u) \diff\tau(t)\\
&= \norm{\nabla_u^\ast \nabla_u g}_{L^2(\tau \otimes \mu)}^2\\
&\qquad -\iint
    (\nabla_u g(t,u))^\top (\nabla_u^2 U(u)) (\nabla_u g(t,u)) \diff\mu_\delta(u) \diff\tau(t).
\end{split}
\end{equation}
We estimate the second term from above.
By computing the Hessian of the energy $U$, we find that the second term equals
\begin{equation}
-\iint \delta^d \sum_{x \in \Lambda} \Bigl[ (V''(u(x)) + m^2) (\partial_{u(x)} g(t,u))^2
        + \abs{\nabla_\delta \partial_{u(x)} g(t,u)}^2 \Bigr]
    \diff\mu_\delta(u) \diff\tau(t).
\end{equation}
We now decompose $h(x)^2 \coloneqq \partial_{u(x)} g(t,u)^2$ as $\sum_{i\in \Z^d} h_i^2$
given by $h_i=\partial_{u(x)} g(t,u) \chi_{i}$,
where $\chi^{2}_i$ is a suitable partition of unity supported on $B(i,2)$ for $i\in \Z^d$.
We will also require a family $\tilde{\chi}_i$ of smooth cutoff functions supported in $B(i,4)$
and equal to $1$ on the support of $\chi_i$,
such that $\sum_{i\in \Z^d} \tilde{\chi}^2_i = 2^d$.
Then we can write the above expression as 
\begin{equation}
-\sum_{i \in \Z^d}\iint \delta^d \sum_{x \in \Lambda} \Bigl[ V''(u(x)) h_i^2
        + m^2 h_i^2 + \abs{\chi_i \nabla_\delta h}^2 \Bigr]
    \diff\mu_\delta(u) \diff\tau(t).
\end{equation}
We may write
\begin{equation}
\chi_i \nabla_\delta h
= \chi_i \nabla_\delta (\tilde\chi_i h)
= \nabla (\chi_i \tilde\chi_i h) - (\nabla \chi_i) \tilde\chi_i h
= \nabla (\chi_i h) - (\nabla \chi_i) \tilde\chi_i h.
\end{equation}
Therefore we can estimate with Young's inequality
\begin{equation}
-\abs{\chi_i \nabla_\delta h}^2
\leq -\frac 1 2 \abs{\nabla_\delta (\chi_i h)}^2
    + \abs{(\nabla \chi_i) \tilde\chi_i h}^2.
\end{equation}
By construction $\abs{\nabla \chi_i} \leq C$ for some uniform, dimension-dependent constant.
Therefore
\begin{equation}
\begin{split}
&\mathrel{\phantom{=}}
-\iint \delta^d \sum_{x \in \Lambda} \Bigl[ V''(u(x)) h_i(x)^2 + m^2 h_i(x)^2
        + \abs{ \chi_i \nabla_\delta h}^2 \Bigr]
    \diff\mu_\delta(u) \diff\tau(t)\\
&\leq -\iint \delta^d \sum_{x \in \Lambda} \Bigl[ V''(u(x)) h_i(x)^2 + m^2 h_i(x)^2
        + \frac{1}{2}\abs{ \nabla_\delta h_i}^2 \Bigr] + C \delta^d \sum_{x \in \Lambda} \tilde{\chi}^{2}_{i}  h^2 \diff\mu_\delta(u) \diff\tau(t).
\end{split}
\end{equation}
Now by Assumption~\ref{ass:potential-intro} this can be bounded by
\begin{equation}\label{eq:clw 2.3 sum}
\sum_{x \in \Lambda} \delta^d \iint (Q_i(u)) (h_i)^2 
    \diff\mu_\delta(u) \diff\tau(t) + C \delta^d \sum_{x \in \Lambda} \tilde{\chi}^{2}_{i}  h^2 \diff\mu_\delta(u) \diff\tau(t)
\end{equation}
Now by Lemma~\ref{thm:ngs bound} (taking $\psi(u) = 2 Q(u)$ and $f(u) = 2^{-1/2} h_i(x)$
with $x$ and $i$ fixed)
and $R^2$ from Assumption~\ref{ass:potential-intro}, the integral is bounded by
\begin{equation}
\int \Bigl[ \delta^d R^2 \int
        (h_i(t,x,u))^2 \diff\mu_\delta(u)
    + \frac 1 2 \int \abs{\nabla_u h_i(t,x,u)}^2 \diff\mu_\delta(u) \Bigr] \diff\tau(t)+C \delta^d \sum_{x \in \Lambda} \tilde{\chi}^{2}_{i}  h^2 \diff\mu_\delta(u) \diff\tau(t)
\end{equation}
Summing over $x \in \Lambda$ and $i\in \Z^d$ we obtain 
\begin{equation}
\begin{split}
&\mathrel{\phantom{=}}
    \norm{\nabla_u^2 g}_{L^2(\tau \otimes \mu)}^2\\
&\leq \norm{\nabla_u^\ast \nabla_u g}_{L^2(\tau \otimes \mu)}^2
    + \delta^d (R^2+2^d C) \iint \abs{\nabla_u g(t,u)}^2 \diff\mu_\delta(u) \diff\tau(t)
    + \frac{1}{2} \norm{\nabla_u^2 g}_{L^2(\tau \otimes \mu)}^{2}.
\end{split}
\end{equation}
Subtraction of the last term on the right-hand side gives the desired bound. 
Under Assumption~\ref{ass:hessian finite} the proof proceeds
without using the partition of unity,
which removes the $2^d C$~term on the right-hand side. 
\end{proof}

With Lemma~\ref{thm:clw 2.3} we may now follow \cite{cao_explicit_2023} to obtain the decay of variance estimate in Theorem~\ref{thm:decay of variance}.

\begin{proof}[Proof of Theorem~\ref{thm:decay of variance}]
Our Assumption~\ref{ass:log-sobolev} is stronger than the
Poincaré inequality assumed in~\cite{cao_explicit_2023}.
This is discussed in Lemma~\ref{thm:clw 2.1},
which replaces~\cite[Lemma~2.1]{cao_explicit_2023}.
We use the NGS~bound in our Lemma~\ref{thm:clw 2.3},
which replaces~\cite[Lemma~2.3]{cao_explicit_2023}.

The only remaining dimension dependency is in~\cite[Lemma~2.2]{cao_explicit_2023},
which uses Item~1 of Assumption~\ref{ass:technical}.
With our proof of Lemma~\ref{thm:clw 2.3} this technical result is only used in
qualitative way at the end of~\cite[Lemma~2.4]{cao_explicit_2023}.
Therefore these modifications already suffice to make the proof uniform in $\delta > 0$.

In~\cite{cao_explicit_2023} $f_0$ is assumed to be in~$L^2(\mu;\, H^1_\kappa)$,
but since the inequality is stated in terms of $L^2$ norm,
we can argue by density that the result holds also for $f_0 \in L^2(\mu \otimes \kappa)$.
\end{proof}

\section{The two-dimensional models}\label{sec:measures}

The goal of this section is to prove Theorem~\ref{thm:integrability}.
Thus we estimate the exponential moment~\eqref{eq:potential assumption exponential}
for the $\Phi_2^4$ and sine-Gordon models.
We begin by introducing the variational approach
to computing expectations of the form $\E \exp(-Q)$.

\subsection{The variational approach}

Consider a space-time white noise $X$ on $\mathbb{R} \times \Lambda$,
and denote by $\Prob$ its law on $\mathcal{S}'([0,T]\times \Lambda)$.
We define the space of progressively measurable processes (drifts)
\begin{equation}
\mathbb{H}_a = \left\{ h \text{ adapted to } X
    \colon h \in L^2([0,T]\times \Lambda) \text{ almost surely} \right\}.
\end{equation}
Let $\mathcal{P}$ denote the space of probability measures on
$\mathcal{S}'([0,T]\times \Lambda)\times L^2([0,T]\times \Lambda)$.
We set
\begin{equation}
\begin{split}
\mathcal{X}
&= \left\{ \zeta \in \mathcal{P} \colon
    \exists h \in \mathbb{H}_a \text{ such that }
    \zeta = \operatorname{Law}(X,h) \right\},\\
\overline{\mathcal{X}}
&= \left\{ \zeta \in \mathcal{P} \colon
    \zeta_n \to \zeta \text{ weakly for some } \zeta_n\in\mathcal X
    \text{ with }
    \sup_{n\in\mathbb{N}}\E_{\zeta_n}
    \norm{h}^{2}_{L^2([0,T]\times\Lambda)}<\infty
    \right\}.
\end{split}
\end{equation}
We will define the regularizing operator
\begin{equation}
J_t = (1-\Laplace_\delta)^{-1/2} \I_{t \in [0,1]}
\end{equation}
as Fourier multiplier, and
\begin{equation}
I_t(h) = \int^t_0 J_s h_s \diff s,\qquad
W_t = \int^t_0 J_s \diff X_s.
\end{equation}
In case $t=1$ we will drop the subscript and write $I=I_1(h)$, $W=W_1$.
In fact we use this parametrization only for the $\Phi^4_2$ measure.
We define a different regularization for the sine-Gordon construction
in Section~\ref{sec:sine-gordon measure}.
In the sequel we denote by $\nu_{\delta}=\operatorname{Law}(W_1)$
the Gaussian free field on $\Lambda_\delta$ with mass $1$.
It will be convenient to use mass $1$ for technical reasons.
Since this is the reference measure, we add a term
$\sum_{x\in\Lambda_{\delta}}\phi^2$ to the potential.

\begin{theorem}[Boué--Dupuis formula]\label{thm:boue-dupuis}
If $F \colon \R^{\Lambda} \to \R$ is continuous and bounded below,
then
\[
\begin{split}
-\log \int \exp(-F(u)) \diff \nu_{\delta}(u)
&= \inf_{h \in \mathbb H_{a}} \E_{\mathbb{P}} \left[
    F(W + I[h]) + \frac 1 2 \int_0^1 \norm{h_s}_{L^2}^2 \diff s
    \right]\\
&= \min_{\zeta \in \overline{\mathcal X}} \E_\zeta \left[
    F(W + I[h]) + \frac 1 2 \int_0^1 \norm{h_s}_{L^2}^2 \diff s
    \right].
\end{split}
\]
\end{theorem}
\begin{proof}
Both versions can be found in \cite[Lemmas~2 and~8]{barashkov_variational_2023}.
\end{proof}

We take $Q(u)$ to be
\begin{equation}\label{eq:Q}
Q(u)
= -\inf_{\norm{\psi}_{L^2(\Lambda)} = 1}
  \delta^2 \sum_{x\in\Lambda} 
  \left[
    \chi(x) V''(u(x))\psi(x)^2
    + \frac 1 2 \abs{\nabla_\delta \psi(x)}^2
    + m^2 \psi(x)^2
  \right].
\end{equation}
with positive $\chi \in C^\infty_c$ with support in
$B(0,3)$ such that $\chi\equiv 1$ on $B(0,2)$.
Observe that $Q$ is the bottom of the spectrum of a discrete Schrödinger operator. 
Then by construction we have
\begin{equation}
\delta^2 \sum_{x \in \Lambda}
  \chi \Bigl[V''(u)\psi^2
  + \frac 1 2 \abs{\nabla_\delta \psi}^2
  + m^2\psi^2 \Bigr]
\geq - Q(u) \norm{\psi}^2_{L^2(\Lambda)}.
\end{equation}

We can now rewrite the left-hand side of \eqref{eq:integrability bounds} as
\begin{equation}\label{eq:log-diff}
\begin{split}
&\mathrel{\phantom{=}}
    \log \int e^{q Q(u)} \diff\mu_\delta(u)\\
&= \log \int e^{q Q(u) - \sum_{x \in \Lambda}[ V(u)-(1-m^2):u^2:]} \diff\nu_\delta(u)
    - \log \int e^{-\sum_{x \in \Lambda}[ V(u)-(1-m^2):u^2:]} \diff\nu_\delta(u).
\end{split}
\end{equation}
This form is suitable for approximation with the Boué--Dupuis formula.

To treat the $\phi^4$ models, we will need the following characterisation of  the gradient of $Q(u)$ in terms of an optimization problem.

\begin{lemma}[Hellman-Feynman formula]\label{lem:Grad-bottom}
Assume $V \in C^3(\R)$.
For any $\delta > 0$, we have
\[
-\partial_{u(x)} Q(u) = \chi(x) V'''(u(x))(\psi^u(x))^2,
\]
where $\psi^u$ of unit $L^2(\Lambda)$ norm minimizes
\[
    \psi \mapsto
    \sum_{x \in\Lambda}
    \left(\chi V''(u)\psi^2
    +m^2\psi^2+ \frac 1 2 \abs{\nabla_\delta \psi}^2\right).
\]
\end{lemma}
\begin{proof}
The gradient of $Q(u)$ is well-defined by Kato perturbation theory.
Indeed Q is the bottom eigenvalue of a discrete Schrödinger operator,
which is known to be simple by comparison principle,
so Kato perturbation theory implies that $Q$ is smooth if $V$ depends smoothly on $u$~\cite[Chapter~2]{kato_1996}. 
For any $v \in L^2(\Lambda)$ and $t > 0$,
we can plug $\psi^u$ into both infima to estimate
\begin{equation}
\begin{split}
&\mathrel{\phantom{=}} \frac{1}{t} \inf_{\norm{\psi}_{L^2(\Lambda)} = 1}
    \left( \sum_{x \in\Lambda} \chi V'' (u + t v) \psi^2
        + \frac 1 2 \abs{\nabla_\delta \psi}^2
        + m^2 \psi^2  \right)\\
&\qquad - \frac{1}{t} \inf_{\norm{\psi}_{L^2(\Lambda)} = 1}
    \left( \sum_{x \in\Lambda} \chi V'' (u) \psi^2
        + \frac 1 2 \abs{\nabla_\delta \psi}^2
        + m^2 \psi^2   \right)\\
&\leq \frac{1}{t} \left( \sum_{x \in\Lambda} \chi V''(u + t v) (\psi^u)^2 
    - \sum_{x \in\Lambda} \chi V''(u) (\psi^u)^2  \right)\\
&\to \sum_{x \in\Lambda} \chi V'''(u)(\psi^u)^2 v
\end{split}
\end{equation}
in the $t \to 0$ limit.
We obtain the converse inequality by replacing $v$ with $-v$.
Therefore
\begin{equation}
-\langle \nabla_u Q (u), v \rangle_{L^2 (\Lambda)}
= \langle \chi V''' (u) (\psi^u)^2, v \rangle_{L^2 (\Lambda)},
\end{equation}
from which the statement follows by duality. 
\end{proof}
In the following we view  $\nabla Q(u)$ as a function given by $\nabla Q(u)(x)=\partial_{u(x)}Q(u)$.

\subsection{The \texorpdfstring{$\Phi^4_2$}{Phi42} model}\label{sec:phi42 measure}

In this case $Q$ is given by
\begin{equation}
Q(u) = -\inf_{\norm{\psi}_{L^2} = 1} \delta^2 \sum_{x \in \Lambda} \left[
    12\lambda \chi(x) \wick{u(x)^2} \psi(x)^2
    + \frac 1 2 \abs{\nabla_\delta \psi(x)}^2 + m^2 \psi(x)^2
\right].
\end{equation}
By \eqref{eq:log-diff} and Theorem~\ref{thm:boue-dupuis} we can write
\begin{equation}
\begin{split}
&\mathrel{\phantom{=}}
\log \int e^{q Q(u)} \diff\mu_{\delta}\\
&= -\min_{\zeta \in \overline{\mathcal{X}}} \E_\zeta \left[
    -q Q(W + I[h])
    + \lambda \sum_{x \in \Lambda} \wick{(W + I[h])^4}
    - (1-m^2)\sum_{x \in \Lambda}\wick{(W+I[h])^2}
    + \frac 1 2 \int_0^1 \norm{h_s}_{L^2}^2 \diff s \right]\\
&\qquad + \min_{\zeta \in \overline{\mathcal{X}}} \E_\zeta \left[
    \lambda \sum_{x \in \Lambda} \wick{(W + I[h])^4}
    - (1-m^2)\sum_{x\in \Lambda} \wick{(W+I[h])^2}
    + \frac 1 2 \int_0^1 \norm{h_s}_{L^2}^2 \diff s \right].
\end{split}
\end{equation}
If $\zeta^Q$ minimizes the first term
\begin{equation} \label{eq:functional Q}
\zeta \mapsto \E_\zeta \left[
    -q Q(W + I[h])
    + \lambda \sum_{x \in \Lambda} \wick{(W + I[h])^4}
    - (1-m^2)\sum_{x\in \Lambda} \wick{(W+I[h])^2}
    + \frac 1 2 \int_0^1 \norm{h_s}_{L^2}^2 \diff s \right],
\end{equation}
then we can again plug it into both minima to bound
\begin{equation}\label{eq:quantity}
\log \int e^{q Q(u)} \diff\mu_{\delta}
\leq q \E_{\zeta^Q} Q(W + I[h]) .
\end{equation}

Before we proceed let us briefly sketch the strategy to estimate~\eqref{eq:quantity}.
Note by definition of~$Q$ it suffices to bound from below 
\begin{equation}\label{eq:lower-bound-Q}
12\lambda\sum_{x\in\Lambda} \chi\bigl(\wick{W^2} + 2WI + I^2\bigr)(\psi^{W+I})^2
\geq 12\lambda\sum_{x\in\Lambda} \chi\bigl(\wick{W^2} + 2WI\bigr)(\psi^{W+I})^2 .
\end{equation}
To bound the right-hand side,
we first derive an \emph{a priori} bound on~$\psi^{W+I}$ in terms of~$I$
and Wick powers of~$W$ in Lemma~\ref{lem:estimate-psi-phi42}.
We then bound~$I$ in terms of the Wick powers of~$W$
in a few steps culminating in Lemma~\ref{lemma:apriori}.
This gives a bound that only depends on $W$ and its powers, and the model parameters.

To achieve the \emph{a priori} bound on~$I$ uniformly in the volume,
we will use an Euler--Lagrange type equation for the minimizer of~\eqref{eq:functional Q}.
We equip $\R^2$ with a weight
$\rho(x) \coloneqq (1+ K_2\abs{x}^2)^{-K_1/2}$.
We will choose $K_1 > 0$ to be large enough that all fractional powers of~$\rho$ appearing in the proof are integrable,
and $K_2 > 0$ small enough to satisfy~\eqref{eq:weight derivative}.

We observe that the first variation of the functional~\eqref{eq:functional Q}
at $\zeta^Q$, in the direction
$\zeta^{Q,t}=\operatorname{Law}_{\zeta^Q}(W,h+t\rho h)$, is
\begin{equation}
\begin{split}\label{eq:el phi4}
0
&= \E_{\zeta^Q} \left[
    -q \delta^2\sum_{x \in \Lambda} \rho \nabla Q(W + I) I
    + 4\lambda \delta^2\sum_{x \in \Lambda} \rho \wick{(W + I)^3} I
    -2(1-m^2) \delta^2\sum_{x \in \Lambda}\rho(W+I)I
    \right]\\
&\quad + \E_{\zeta^Q}\left[
    \int_0^1 \delta^2\sum_{x \in \Lambda}
        \rho \abs{\nabla_{\delta}\dot{I}_s}^2 \diff s
    + \int_0^1 \delta^2\sum_{x\in \Lambda}
        \rho \dot{I}_s^2 \diff s
    + \int_0^1 \delta^2\sum_{x\in \Lambda}
        ([\rho,\nabla_{\delta}]\dot{I}_s) \nabla_{\delta}\dot{I}_s
        \diff s
    \right]\!,
\end{split}
\end{equation}
See \cite[Lemma~13]{barashkov_variational_2023} for further details.
Recalling the computation of $\nabla Q$ from Lemma~\ref{lem:Grad-bottom},
we thus obtain
\begin{equation}\label{eq:EL-phi42}
\begin{split}
0
&= \E_{\zeta^Q} \left[
    24 \lambda q \delta^2\sum_{x \in \Lambda}
        \rho \chi (\psi^{W + I})^2 (W + I) I
    + 4\lambda \delta^2\sum_{x \in \Lambda}
        \rho \wick{(W + I)^3} I \right.\\
&\hspace{5em}\left.
    - 2(1-m^2) \delta^2\sum_{x \in \Lambda}\rho(W+I)I
    + \int_0^1 \delta^2\sum_{x \in \Lambda}
        \rho \abs{\nabla_{\delta}\dot{I}_s}^2 \diff s
    \right]\\
&\quad + \E_{\zeta^Q}\left[
    \int_0^1 \delta^2\sum_{x\in \Lambda}\rho \dot{I}_s^2 \diff s
    + \int_0^1 \delta^2\sum_{x\in \Lambda}
        ([\rho,\nabla_{\delta}]\dot{I}_s) \nabla_{\delta}\dot{I}_s
        \diff s
    \right]\!.
\end{split}
\end{equation}

As discussed above, let us first show an \emph{a priori} bound on $\psi^{W+I}$:

\begin{lemma}\label{lem:estimate-psi-phi42}
The following estimate holds:
\begin{equation}
\begin{split}
\norm{\psi^{W+I}}_{H^1(\Lambda)}^2
&\leq C_\varepsilon \lambda
    \left(
      1+\norm{\chi \wick{W^2}}_{\Holder^{-\varepsilon}(\Lambda)}
       +\norm{\chi^{1/2}W}_{\Holder^{-\varepsilon}(\Lambda)}
        \norm{\chi^{1/2}I}_{H^{1}(\Lambda)}
    \right)^{1/(1-\varepsilon)}\\
&\quad + 2 \lambda\norm{\chi^{1/2}I}_{L^2(\Lambda)}^2 + 2.
\end{split}
\end{equation}
\end{lemma}
\begin{proof}
Let $\psi=\psi^{W+I}$ be as in Lemma~\ref{lem:Grad-bottom},
minimizing the quadratic form
\begin{equation}\label{eq:estimate-psi-phi42 qf}
12 \lambda \delta^2 \sum_{x\in\Lambda}\chi \wick{(W+I)^2}\psi^2
    + \frac 1 2 \norm{\nabla_\delta \psi}_{L^2}^2
    + m^2 \norm{\psi}_{L^2}^2.
\end{equation}
Since $\norm{\psi}_{L^2}^2=1$ by assumption,
\begin{equation}
\norm{\nabla_\delta \psi}_{L^2}^2
= \norm{\psi}_{H^1}^2 - 1.
\end{equation}
Moreover by duality and interpolation,
\begin{equation}\label{eq:estimate-psi-phi42 w terms}
\begin{split}
&\mathrel{\phantom{=}}
    12\lambda \delta^2 \bigabs{\sum_{x\in\Lambda}\chi \wick{W^2} \psi^2}
    + 24\lambda \delta^2 \bigabs{\sum_{x\in\Lambda}\chi WI\psi^2}\\
&\leq C\lambda
    \left(\norm{\chi \wick{W^2}}_{\Holder^{-\varepsilon}}
        + \norm{\chi^{1/2}W}_{\Holder^{-\varepsilon}}\norm{\chi^{1/2}I}_{H^{1}} \right)
    \norm{\psi}_{H^\varepsilon}^2.
\end{split}
\end{equation}
Let us set
\begin{equation}
A(W,I) = 1 + \norm{\chi \wick{W^2}}_{\Holder^{-\varepsilon}}
    + \norm{\chi^{1/2}W}_{\Holder^{-\varepsilon}}
      \norm{\chi^{1/2}I}_{H^{1}}.
\end{equation}
Therefore by Young's inequality
\begin{equation}
\begin{split}
\eqref{eq:estimate-psi-phi42 w terms}
&\leq C\lambda A(W,I) \norm{\psi}_{H^1}^{2\varepsilon}\\
&\leq \frac{1}{4} \norm{\psi}_{H^1}^2
    + C_\varepsilon \lambda^{1/(1-\varepsilon)} A(W,I)^{1/(1-\varepsilon)},
\end{split}
\end{equation}
and we get the lower bound
\begin{equation}
\eqref{eq:estimate-psi-phi42 qf}
\geq \frac{1}{4} \norm{\psi}_{H^1}^2
    - C_\varepsilon \lambda^{1/(1-\varepsilon)} A(W,I)^{1/(1-\varepsilon)}
    - 1 + m^2.
\end{equation}

At the same time we can bound~\eqref{eq:estimate-psi-phi42 qf} from above
by plugging some fixed $\psi = \psi_0$ into the quadratic form.
Let us choose $\psi_0 = \abs{\Tor^2}^{-1/2}$.
Then $\norm{\psi_0}_{L^2(\Lambda)} = 1$ and $\psi_0 \leq 1$, so
\begin{equation}
\begin{split}
&\mathrel{\phantom{=}}
    12 \lambda \delta^2 \sum_{x \in \Lambda} \chi \wick{(W + I)^2} \psi_0^2
    + \frac 1 2 \norm{\nabla_\delta \psi_0}_{L^2}^2
    + m^2 \norm{\psi_0}_{L^2}^2\\
&= \frac{\lambda \delta^2}{|\Tor^2|} \sum_{x \in \Lambda} \chi \wick{(W + I)^2} + m^2\\
&\leq \lambda \norm{\chi \wick{W^2}}_{\Holder^{- \varepsilon}}
    + C\lambda \norm{\chi^{1 / 2} W}_{\Holder^{- \varepsilon}}
        \norm{\chi^{1 / 2} I}_{H^{\varepsilon}}
    + \lambda \norm{\chi^{1 / 2} I}^2_{L^2}
    + m^2\\
&\leq C\lambda A(W,I) + \lambda \norm{\chi^{1/2}I}_{L^2}^2 + m^2.
\end{split}
\end{equation}
Here we expanded the Wick power and used duality.
Combining these two estimates gives
\begin{equation}
\frac{1}{4} \norm{\psi}_{H^1}^2
\leq C\lambda A(W,I) + C_\varepsilon \lambda^{1/(1-\varepsilon)} A(W,I)^{1/(1-\varepsilon)}
    + \lambda \norm{\chi^{1/2}I}_{L^2}^2 + 1.
\end{equation}
Since $A(W,I) \geq 1$ and $\lambda \leq 1$, this yields the statement.
\end{proof}

The following lemmas then estimate each term in~\eqref{eq:EL-phi42}
to find a bound on~$I$:

\begin{lemma}\label{Lem:gradient-quadratic-form}
For $\varepsilon>0$ small enough and $\kappa>0$, we have
\begin{equation*}
24 \lambda q \delta^2 \sum_{x \in \Lambda} \chi(\psi^{W+I})^2(W+I)I
\geq -\mathcal{G}_{q,\lambda,m,\varepsilon,\kappa}(W)
    -\kappa m^2\norm{I}^2_{H^{1}(\rho^{1/2})},
\end{equation*}
where
\[
\E \abs{\mathcal{G}_{q,\lambda,m,\varepsilon,\kappa}(W)}
\leq C_{\kappa,\varepsilon} \lambda^2 \left[
    q^2
    + q^{2/(1-4\varepsilon)}
\right] (1 + \abs{\log(m)}^{p}).
\]
in the regime~$\lambda \leq 1$ for some $p$ dependent only on $\varepsilon$.
\end{lemma}
\begin{proof}
We begin by observing that
$q \lambda \int \chi(\psi^{W+I})^2 I^2 \dx \geq 0$.
We now estimate the
remaining term
\begin{equation}\label{eq:estimate-product-WI}
\begin{split}
&\mathrel{\phantom{=}} q\lambda\left|\delta^2 \sum_{x \in \Lambda}
    \chi(\psi^{W+I})^2WI\right|\\
&\leq C q\lambda  \norm{W}_{\Holder^{-\varepsilon}(\rho)}
    \norm{\chi(\psi^{W+I})}^{2}_{H^{\varepsilon}(\Lambda)} \norm{I}_{H^{1}(\rho^{1/2})}\\
&\leq C q\lambda  \norm{W}_{\Holder^{-\varepsilon}(\rho)} \left[
    \lambda \left(1+\norm{\chi \wick{W^2}}_{\Holder^{-\varepsilon}}
        +\norm{\chi^{1/2} W}_{\Holder^{-\varepsilon}} \norm{\chi^{1/2}I}_{H^{1}}
    \right)^{1/(1-\varepsilon)} \right.\\
&\hspace{8em} \left.{} + 2\lambda\norm{\chi^{1/2}I}_{L^2}^2 + 2 \right]^{2\varepsilon}
        \norm{I}_{H^{1}(\rho^{1/2})} .
\end{split}
\end{equation}
In the last line we used Lemma~\ref{lem:estimate-psi-phi42} and Besov interpolation.
To continue the estimate, set
\begin{equation}
    X=\norm{I}_{H^1(\rho)},\qquad
    \alpha=\frac{2\varepsilon}{1-\varepsilon}.
\end{equation}
Using concavity of $x\mapsto x^{2\varepsilon}$ and the bound
$\norm{\chi^{1/2}I}_{L^2(\Lambda)}+ \norm{\chi^{1/2}I}_{H^1(\Lambda)}\lesssim X$,
equation~\eqref{eq:estimate-product-WI} is bounded by
\begin{equation}
Cq\lambda \norm{W}_{\Holder^{-\varepsilon}(\rho)}
\left[
    \lambda^{2\varepsilon}
        \!\left(1 +\norm{\chi \wick{W^2}}_{\Holder^{-\varepsilon}(\rho)}^\alpha \right)\! X
    + \lambda^{2\varepsilon} \norm{\chi^{1/2} W}_{\Holder^{-\varepsilon}(\rho)}^\alpha X^{1+\alpha}
    + \lambda^{2\varepsilon} X^{1+4\varepsilon}
    + X
\right]\!.
\end{equation}
We extract~$X^2$ from each term with Young's inequality (which requires $\varepsilon<1/4$).
This gives
\begin{equation}
\begin{split}
&\mathbin{\phantom{+}}
    C_\kappa q^2 \lambda^{2+4\varepsilon} \norm{W}_{\Holder^{-\varepsilon}(\rho)}^2
    \!\left(1 +\norm{\chi \wick{W^2}}_{\Holder^{-\varepsilon}(\rho)}^\alpha \right)^2\\
&+ C_\kappa q^{2/(1-\alpha)} \lambda^{2(1+2\varepsilon)/(1-\alpha)}
    \norm{W}_{\Holder^{-\varepsilon}(\rho)}^{2/(1-\alpha)}
    \norm{\chi^{1/2} W}_{\Holder^{-\varepsilon}(\rho)}^{2\alpha/(1-\alpha)}\\
&+ C_\kappa q^{2/(1-4\varepsilon)} \lambda^{2(1+2\varepsilon)/(1-4\varepsilon)}
    \norm{W}_{\Holder^{-\varepsilon}(\rho)}^{2/(1-4\varepsilon)}\\
&+ C_\kappa q^2 \lambda^2 \norm{W}_{\Holder^{-\varepsilon}(\rho)}^2\\
&+ \kappa m^2 X^2.
\end{split}
\end{equation}
This is the decomposition in the statement.
To compute expectations of the stochastic terms, we use Lemma~\ref{wick-power-estimate}.
It gives $p$ dependent on $\varepsilon$ such that
each expectation is at most $C_\varepsilon (1 + \abs{\log(m)}^{p})$.
Hence $\E \abs{\mathcal{G}_{q,\lambda,m,\varepsilon,\kappa}(W)}$ is bounded by
\begin{equation}
C_{\kappa,\varepsilon} \left[
    q^2 \lambda^{2+4\varepsilon}
    + q^{2/(1-\alpha)} \lambda^{2(1+2\varepsilon)/(1-\alpha)}
    + q^{2/(1-4\varepsilon)} \lambda^{2(1+2\varepsilon)/(1-4\varepsilon)}
    + q^2 \lambda^2
\right] (1 + \abs{\log(m)}^{p}).
\end{equation}
Under the assumption $\lambda \leq 1$ we may use $1-\alpha \leq 1$
to further simplify this expression.
\end{proof}

\begin{lemma}\label{Lem:coercive}
For any $p \geq 8$, there exists $C<\infty$ such that
\begin{equation}\label{eq:coercive-with-mass-shift}
\begin{split}
&\mathrel{\phantom{=}}
    4\lambda \delta^2 \sum_{x \in \Lambda} \rho \wick{(W + I)^3} I
    - 2(1-m^2) \delta^2 \sum_{x \in \Lambda} \rho (W+I) I\\
&\geq 
    3\lambda \norm{I}^4_{L^4 (\rho^{1/4})}
    - m^2 \norm{I}^2_{H^1 (\rho^{1/2})}
    - C\frac{(1-m^2)^2}{\lambda} \norm{\rho}_{L^1(\Lambda)}\\
&\qquad - C \bigg(\frac{\lambda^2}{m^2}
        + \frac{\lambda^3}{m^4}
        + \frac{(1-m^2)^4}{m^2}\bigg)
    \!\bigg(1 + \sum_{i = 1}^3 \norm{\wick{W^i}}_{\Holder^{-\varepsilon}(\rho^{1/8})}^p\bigg).
\end{split}
\end{equation}
\end{lemma}
\begin{proof}
By the Wick binomial formula,
\begin{equation}
    4\wick{(W + I)^3} I = 4I^4 + 12 W I^3 + 12 \wick{W^2} I^2 + 4 \wick{W^3} I .
\end{equation}
The first term gives $4 \lambda \norm{I}_{L^4(\rho^{1/4})}^4$.
The three mixed terms match the last three terms
in the proof of~\cite[Lemma~21]{barashkov_variational_2023}.
By duality and the weighted product estimates, for $\varepsilon>0$ small enough,
\begin{equation}
\begin{split}
\bigabs{\delta^2 \sum_{x \in \Lambda} \rho \wick{W^3} I \dx}
    &\leq C \norm{\wick{W^3}}_{\Holder^{-\varepsilon}(\rho^{1/8})}
        \norm{I}_{H^1(\rho^{1/2})},\\
\bigabs{\delta^2 \sum_{x \in \Lambda} \rho \wick{W^2} I^2 \dx}
    &\leq C \norm{\wick{W^2}}_{\Holder^{-\varepsilon}(\rho^{1/8})}
        \norm{I}_{L^4(\rho^{1/4})} \norm{I}_{H^1(\rho^{1/2})},\\
\bigabs{\delta^2 \sum_{x \in \Lambda} \rho W I^3 \dx}
    &\leq C \norm{W}_{\Holder^{-\varepsilon}(\rho^{1/8})}
        \norm{I}_{H^1(\rho^{1/2})}^{1/2}
        \norm{I}_{L^4(\rho^{1/4})}^{5/2}.
\end{split}
\end{equation}
Applying Young's inequality with the weights
$\lambda\norm{I}_{L^4(\rho^{1/4})}^4$ and
$m^2\norm{I}_{H^1(\rho^{1/2})}^2$ gives
\begin{equation}
\begin{split}
\lambda \norm{\wick{W^3}}_{\Holder^{-\varepsilon}(\rho^{1/8})}
        \norm{I}_{H^1(\rho^{1/2})}
    &\leq \frac{m^2}{4}\norm{I}_{H^1(\rho^{1/2})}^2
        + C\frac{\lambda^2}{m^2}
        \norm{\wick{W^3}}_{\Holder^{-\varepsilon}(\rho^{1/8})}^2,\\
\lambda \norm{\wick{W^2}}_{\Holder^{-\varepsilon}(\rho^{1/8})}
        \norm{I}_{L^4(\rho^{1/4})} \norm{I}_{H^1(\rho^{1/2})}
    &\leq \frac{\lambda}{4}\norm{I}_{L^4(\rho^{1/4})}^4
        + \frac{m^2}{4}\norm{I}_{H^1(\rho^{1/2})}^2
        + C\frac{\lambda^3}{m^4}
        \norm{\wick{W^2}}_{\Holder^{-\varepsilon}(\rho^{1/8})}^4,\\
\lambda \norm{W}_{\Holder^{-\varepsilon}(\rho^{1/8})}
        \norm{I}_{H^1(\rho^{1/2})}^{1/2}
        \norm{I}_{L^4(\rho^{1/4})}^{5/2}
    &\leq \frac{\lambda}{4}\norm{I}_{L^4(\rho^{1/4})}^4
        + \frac{m^2}{4}\norm{I}_{H^1(\rho^{1/2})}^2
        + C\frac{\lambda^3}{m^4}
        \norm{W}_{\Holder^{-\varepsilon}(\rho^{1/8})}^8.
\end{split}
\end{equation}
It remains to bound the quadratic term of the statement.
For the mixed linear term, duality and Young's inequality give
\begin{equation}
\begin{split}
(1-m^2) \bigabs{\delta^2 \sum_{x\in\Lambda}\rho W I}
&\leq C (1-m^2) \norm{W}_{\Holder^{-\varepsilon}(\rho^{1/2})}
    \norm{I}_{H^1(\rho^{1/2})}\\
&\leq \frac{m^2}{4} \norm{I}_{H^1(\rho^{1/2})}^2
    + \frac{C(1-m^2)^4}{m^2} \norm{W}^{4}_{\Holder^{-\varepsilon}(\rho^{1/8})} .
\end{split}
\end{equation}
For the quadratic we use the pointwise estimate
\begin{equation}
    (1-m^2)\,|I|^2
    \leq \frac{\lambda}{2} |I|^4
        + C\frac{(1-m^2)^2}{\lambda},
\end{equation}
which after multiplying by $\rho$ and summing yields
\begin{equation}
    -(1-m^2)\delta^2 \sum_{x\in\Lambda}\rho I^2
    \geq -\frac{\lambda}{2}\norm{I}_{L^4(\rho^{1/4})}^4
        -C\frac{(1-m^2)^2}{\lambda} \norm{\rho}_{L^1(\Lambda)}.
\end{equation}
The previous mixed cubic estimates used only $\lambda/2$ of the leading
quartic term and $3m^2/4$ of the $H^1$ term. The two estimates above use
another $\lambda/2$ and $m^2/4$, leaving precisely the coercive terms displayed
in~\eqref{eq:coercive-with-mass-shift}. Choosing $p \geq 8$ and using
$a^2+a^4+a^8 \leq C(1+a^p)$ completes the proof.
\end{proof}

\begin{lemma}
\label{lemma:apriori}
For some $C$, $p<\infty$ independent of $\lambda$, $m$, $q$, and~$\delta$,
we have
\[
    3\lambda \E_{\zeta^Q}
      \norm{I}^4_{L^4(\rho^{1/4})}
    + (1 - (1+\kappa) m^2)\E_{\zeta^Q}
      \norm{I}^2_{H^1(\rho^{1/2})}
    \leq C\mathcal{B}_{q,\lambda,m,\varepsilon},
\]
where
\[
\mathcal{B}_{q,\lambda,m,\varepsilon}
\coloneqq
\bigg(\lambda^2 q^2 + \lambda^2 q^{2/(1-4\varepsilon)}
    + \frac{\lambda^2}{m^2}
    + \frac{\lambda^3}{m^4}
    + \frac{(1-m^2)^4}{m^2}\bigg) (1 + \abs{\log(m)}^{p})
+ \frac{(1-m^2)^2}{\lambda} \norm{\rho}_{L^1(\Lambda)}
\]
when $\lambda \leq m < 1$.
\end{lemma}
\begin{proof}
Recalling \eqref{eq:EL-phi42},
\begin{equation}\label{eq:EL-phi42-repeat}
\begin{split}
0
&= \E_{\zeta^Q} \left[
    24 \lambda q \delta^2 \sum_{x \in \Lambda}
        \rho \chi (\psi^{W + I})^2 (W + I) I
    + 4\lambda \delta^2 \sum_{x \in \Lambda}
        \rho \wick{(W + I)^3} I \right.\\
&\hspace{5em}\left.
    - 2(1-m^2) \delta^2 \sum_{x \in \Lambda}\rho(W+I)I
    + \int_0^1 \delta^2 \sum_{x \in \Lambda}
        \rho \left[ \abs{\nabla_{\delta}\dot{I}_s}^2 + \dot{I}_s^2 \right] \diff s
    \right]\\
&\quad + \E_{\zeta^Q}
    \int_0^1 \delta^2 \sum_{x\in \Lambda}
        ([\rho,\nabla_{\delta}]\dot{I}_s) \nabla_{\delta}\dot{I}_s
        \diff s.
\end{split}
\end{equation}
Let us first fix the weight.
Since $\rho(x)=(1+ K_2\abs{x}^2)^{-K_1/2}$, we can estimate for example
\begin{equation}
\abs{\partial_{x_i} \rho(x)}
= \bigabs{K_1 K_2 \rho(x) \frac{x_i}{1 + K_2 \abs x^2}}
\leq K_1 K_2^{1/2} \abs{\rho(x)}.
\end{equation}
We can therefore choose $K_2 > 0$ to be small enough (dependent on $K_1$) that
\begin{equation}\label{eq:weight derivative}
    \abs{\nabla_{\delta}\rho(x)}
    \leq \frac 3 2 \abs{\rho(x)}
\end{equation}
uniformly in $\delta \lesssim 1$.
Hence by Cauchy--Schwarz
\begin{equation}
\bigabs{\delta^2 \sum_{x\in\Lambda}
    ([\rho,\nabla_\delta]\dot{I})(x)\nabla_\delta \dot{I}(x)}
\leq \frac{3\delta^2}{4} \sum_{x\in\Lambda}\rho(x) \abs{\nabla_\delta \dot{I}(x)}^2
    + \frac{3\delta^2}{4} \sum_{x\in\Lambda}\rho(x) \dot{I}(x)^2 .
    \label{eq:weight-commutator}
\end{equation}
Taking~$\varepsilon$ sufficiently small, we therefore have
\begin{equation}
\int_{0}^1 \delta^2 \sum_{x\in \Lambda} \bigl[ ([\rho,\nabla_{\delta}]\dot I_s)\nabla_\delta \dot I_{s}
    + \rho\abs{\nabla_{\delta}\dot I_s}^2+ \rho \dot{I}_s^2 \bigr] \diff s
\geq \frac{1}{4} \int_0^1 \delta^2 \sum_{x \in \Lambda}
    \rho \bigl[ \abs{\nabla_{\delta}\dot I_s}^2 + \dot{I}_s^2 \bigr] \diff s.
\end{equation}
Then by Jensen's inequality we may bound
\begin{equation}
\int_0^1 \delta^2 \sum_{x \in \Lambda}
    \rho \bigl[ \abs{\nabla_{\delta}\dot I_s}^2+ \dot{I}_s^2 \bigr] \diff s
\geq \norm{I_{1}}_{H^1(\rho^{1/2})}^2.
\end{equation}
Therefore the last two terms of~\eqref{eq:EL-phi42-repeat}
can be controlled with $\norm{I}_{H^1(\rho^{1/2})}$.
Now the first term of~\eqref{eq:EL-phi42-repeat}
is estimated by Lemma~\ref{Lem:gradient-quadratic-form} as
\begin{equation}
24 \lambda q \delta^2 \E \sum_{x \in \Lambda} \chi(\psi^{W+I})^2(W+I)I
\geq -C_{\kappa,\varepsilon} \lambda^2 \left[ q^2 + q^{2/(1-4\varepsilon)}\right]
        (1 + \abs{\log(m)}^{p})
    -\kappa m^2 \E \norm{I}^2_{H^{1}(\rho^{1/2})},
\end{equation}
and the next two terms by Lemma~\ref{Lem:coercive} and Lemma~\ref{wick-power-estimate} as
\begin{equation}
\begin{split}
&\mathrel{\phantom{=}}
    4\lambda \delta^2 \E \sum_{x \in \Lambda} \rho \wick{(W + I)^3} I
    - 2(1-m^2) \delta^2 \E \sum_{x \in \Lambda} \rho (W+I) I\\
&\geq 
    3\lambda \E \norm{I}^4_{L^4 (\rho^{1/4})}
    - m^2 \E \norm{I}^2_{H^1 (\rho^{1/2})}
    - C\frac{(1-m^2)^2}{\lambda} \norm{\rho}_{L^1(\Lambda)}\\
&\qquad - C \bigg(\frac{\lambda^2}{m^2}
        + \frac{\lambda^3}{m^4}
        + \frac{(1-m^2)^4}{m^2}\bigg) (1 + \abs{\log(m)}^{p})
\end{split}
\end{equation}
for some sufficiently large~$p$.
Therefore we have shown
\begin{equation}
0 \geq 3\lambda \norm{I}_{L^4(\rho^{1/4})}^4
    + (1 - \kappa m^2 - m^2) \E \norm{I}_{H^1(\rho^{1/2})}^2
    - C \mathcal{B}_{q,\lambda,m,\varepsilon},
\end{equation}
which implies the claim once $\kappa$ is taken small enough.
\end{proof}

The following lemma now provides an estimate on~\eqref{eq:quantity}.

\begin{lemma}
\label{estimate-Q-phi4}
Let $q>0$ and let $\varepsilon>0$ be sufficiently small. There exists
$p<\infty$ and a constant $C_{\varepsilon}<\infty$, independent of
$q,\lambda,m,\delta$, such that
\[
q \E_{\zeta^Q} [Q (W + I)]
\leq C_\varepsilon \lambda q (1 + \abs{\log m}^p) \!
    \left[ 1
    + \left(\frac{\mathcal B_{q,\lambda,m,\varepsilon}}{(1 - (1+\kappa) m^2)}\right)^{1/2 + 5\varepsilon}
    \right] \!.
\]
\end{lemma}
\begin{proof}
It suffices to estimate the two terms on the right-hand side of~\eqref{eq:lower-bound-Q} from below.
Duality and interpolation give
\begin{equation}
\lambda q \left|\delta^2 \sum_{x\in\Lambda}\chi \wick{W^2}(\psi^{W+I})^2\right|
\geq -C_\varepsilon\lambda q 
    \norm{\chi \wick{W^2}}_{\Holder^{-\varepsilon}(\rho)}
    \norm{\psi^{W+I}}_{H^1(\Lambda)}^{2\varepsilon}.
\end{equation}
Applying Lemma~\ref{lem:estimate-psi-phi42}, we get
\begin{equation}\label{eq:W2-psi-bound}
\begin{split}
&\mathrel{\phantom{=}} \lambda q
    \norm{\chi \wick{W^2}}_{\Holder^{-\varepsilon}(\rho)}
    \norm{\psi^{W+I}}_{H^1}^{2\varepsilon}\\
&\leq C_\varepsilon \lambda q
    \norm{\chi \wick{W^2}}_{\Holder^{-\varepsilon}(\rho)}
    \Bigl[\lambda
    \left(
      1+\norm{\wick{W^2}}_{\Holder^{-\varepsilon}(\rho)}
       +\norm{W}_{\Holder^{-\varepsilon}(\rho)}
        \norm{I}_{H^{1}(\rho^{1/2})}
    \right)^{1/(1-\varepsilon)}\\
&\hspace{10em} {} + 2 \lambda\norm{I}_{L^2(\rho^{1/2})}^2 + 2\Bigr]^{\varepsilon}.
\end{split}
\end{equation}
Using the concavity of $x\mapsto x^\varepsilon$ and Hölder's inequality
(again assuming $\lambda \leq 1$),
the expectation of this is bounded by
\begin{equation}
\begin{split}
&C_\varepsilon \lambda q \left[\E_{\zeta^Q} \norm{\chi \wick{W^2}}_{\Holder^{-\varepsilon}(\rho)}^2 \right]^{1/2}\\
&\quad
    \left[ 1
        + \E_{\zeta^Q} \norm{\wick{W^2}}_{\Holder^{-\varepsilon}(\rho)}^{2\varepsilon/(1-\varepsilon)}
        + \left( \E_{\zeta^Q} \norm{W}_{\Holder^{-\varepsilon}(\rho)}^{4\varepsilon/(1-\varepsilon)}
            \E_{\zeta^Q} \norm{I}_{H^1(\rho^{1/2})}^{4\varepsilon/(1-\varepsilon)} \right)^{1/2}
        + 2 \E_{\zeta^Q} \norm{I}_{H^1(\rho^{1/2})}^{2\varepsilon}
    \right]^{1/2}.
\end{split}
\end{equation}
By Lemmas~\ref{lemma:apriori} and~\ref{wick-power-estimate},
there exists $p<\infty$ such that this is bounded by
\begin{equation}
C_\varepsilon \lambda q (1 + \abs{\log m}^p) \!
    \left[ 1
    + \left(\frac{\mathcal B_{q,\lambda,m,\varepsilon}}{(1 - (1+\kappa) m^2)}\right)^{\varepsilon/2(1-\varepsilon)}
    + \left(\frac{\mathcal B_{q,\lambda,m,\varepsilon}}{(1 - (1+\kappa) m^2)}\right)^{\varepsilon/2}
    \right].
\end{equation}
Since $\varepsilon$ is small, we can neglect the final term of the bracket.

In the cross term we can estimate by Besov product estimates, interpolation, and Hölder
\begin{equation}
\begin{split}
\lambda q \E \bigabs{\delta^2 \sum_{x\in\Lambda}\chi W I(\psi^{W+I})^2}
&\leq C_\varepsilon \lambda q \norm{W}_{\Holder^{-\varepsilon}(\rho)}
    \norm{\psi^{W+I}}_{H^{2\varepsilon}(\Lambda)}
    \norm{I}_{H^1(\rho^{1/2})}\\
&\leq C_\varepsilon \lambda q \norm{W}_{\Holder^{-\varepsilon}(\rho)}
    \norm{\psi^{W+I}}_{H^1(\Lambda)}^{2\varepsilon}
    \norm{I}_{H^1(\rho^{1/2})}\\
&\leq C_\varepsilon \lambda q \left[
    \E \norm{W}_{\Holder^{-\varepsilon}(\rho)}^4
    \E \norm{\psi^{W+I}}_{H^1(\Lambda)}^{8\varepsilon}
    \right]^{1/4}
    \left[ \E \norm{I}_{H^1(\rho^{1/2})}^2 \right]^{1/2}.
\end{split}
\end{equation}
Again by Lemmas~\ref{lem:estimate-psi-phi42},~\ref{lemma:apriori}, and~\ref{wick-power-estimate}
this is at most
\begin{equation}
C_\varepsilon \lambda q (1 + \abs{\log m}^p) \!
    \left[ 1
    + \left(\frac{\mathcal B_{q,\lambda,m,\varepsilon}}{(1 - (1+\kappa) m^2)}\right)^{4\varepsilon/(1-\varepsilon)}
    \right]
    \left(\frac{\mathcal B_{q,\lambda,m,\varepsilon}}{(1 - (1+\kappa) m^2)}\right)^{1/2}.
\end{equation}
Taking $\varepsilon$ small enough and combining these bounds gives the statement.
\end{proof}

We are now ready to prove Theorem~\ref{thm:integrability} for the $\phi^4_2$ model.
To verify Assumption~\ref{ass:potential-intro}, we need to set $q = 4/\sigma$ in~\eqref{eq:quantity},
where $\sigma$ is the log-Sobolev constant of the measure.
By~\cite{bauerschmidt_logsobolev_2024}, we then have $q \simeq m^{-2}$.
Since we have assumed $m < 1$,
we have
\begin{equation}
\mathcal{B}_{q,\lambda,m,\varepsilon}
\lesssim \frac{\lambda^2}{m^{4/(1-4\varepsilon)}} (1 + \abs{\log m}^p) + \frac{1}{\lambda}.
\end{equation}
If we choose $\lambda \leq m^{4+20\varepsilon}$,
the expression in Lemma~\ref{estimate-Q-phi4} is then bounded uniformly in all parameters.
Together with~\eqref{eq:quantity}, this gives the desired bound on
$\log \int e^{qQ(u)}\diff\mu_\delta(u)$.

\subsection{The sine-Gordon model}\label{sec:sine-gordon measure}

In this section we treat the case of the sine-Gordon model,
where $Q$ is given by
\begin{equation}
Q(u) = -\inf_{\norm{\psi}_{L^2} = 1} \delta^2 \sum_{x \in \Lambda} \left[
    -\beta^2 \lambda\, \chi(x)    \wick{ \cos(\beta u(x))} \psi(x)^2
    +  \frac 12 \abs{\nabla_\delta \psi(x)}^2 + m^2 \psi(x)^2
\right].
\end{equation}
We make some changes to our notation here
in order to use the scale decomposition of the measure.
We replace the regularizing operator by the heat kernel
\begin{equation}
J_t = \left[ \frac{1}{t^2} e^{- (m^2 - \Laplace_\delta) / t} \right]^{1/2},
\end{equation}
and correspondingly parametrize the regularization by $t \in [0, \infty)$.
As before we set
\begin{equation}
I_t(h) = \int^t_0 J_s h_s \diff s,\qquad
W_t = \int^t_0 J_s \diff X_s
\end{equation}
with this new operator.
Again $W_\infty$ has the law~$\nu_\delta$,
the Gaussian free field of covariance $(m^2 - \Laplace_\delta)^{-1}$ on $\Lambda_\delta$.
We will set $W=W_\infty$, $I = I_{\infty}[h]$ for brevity of notation.

\begin{lemma}\label{thm:heat kernel bounds}
The two-dimensional heat kernel $J_t$ satisfies the following estimates:
\begin{enumerate}
\item[(a)] $\displaystyle \bignorm{\int_0^\infty J_t f_t \diff t}_{\Holder^1(\Lambda)}
    \lesssim \sup_{t \geq 0}\, \inorm{t}^{1/2+\varepsilon} \norm{f_t}_{L^\infty(\Lambda)}$,
    where $\varepsilon > 0$.
\item[(b)] $\displaystyle \norm{J_t f}_{L^\infty(\Lambda)} \lesssim \inorm{t}^{-1} \norm{f}_{L^\infty(\Lambda)}$.
\end{enumerate}
Let us denote by~$C_t(x,y)$ and $C^{\ast 2}_t(x,y)$ the convolution kernels
of~$J_t$ and $J_t^2$.
They satisfy:
\begin{enumerate}
\item[(c)] $\norm{C_t}_{L^1(\Lambda)} \lesssim \inorm{t}^{-1}$ independently of~$m$.
\item[(d)] $\norm{C^{\ast 2}_t}_{L^1(\Lambda)} \lesssim m^{-2} \inorm{t}^{-2}$.
\end{enumerate}
\end{lemma}
\begin{proof}
The first estimate is a special case of~\cite[Lemma~A.5]{gubinelli_fbsde_2026},
(b) follows from~(c),
and~(c) and~(d) are direct computations with the convolution kernel representation;
see~\cite[Appendix~A]{bauerschmidt_logsobolev_2021}.
\end{proof}

Our strategy will be different from the $\phi^4$ case. 
We will proceed in the following steps:
First we derive a coupling between the sine-Gordon field and the Gaussian free field. 
Then we will use this coupling to bound $Q(u)$ under the sine-Gordon field,
in terms of an observable under the~GFF.
Lastly we prove exponential tails for that observable. 

\begin{theorem}\label{prop:coupling}
Let $\mu_{\beta}$ be the lattice sine-Gordon measure for $\beta^2 < 4 \pi$,
and $\rho$ a weight as in Section~\ref{sec:phi42 measure}.
There exists a probability space $\Prob$ and random variables
$W \in L^2 (\Prob;\, \Holder^{- \varepsilon}(\Lambda, \rho))$
and $I \in L^{\infty} (\Prob;\, \Holder^1(\Lambda))$ such that
\[
\mu_{\beta} = \operatorname{Law}(W + I), \quad \nu = \operatorname{Law}(W).
\]
such that $\sup_{L} \sup_{\delta} \norm{I}_{L^{\infty} (\Prob, \Holder^1 (\Lambda))}
\leq A(m,\lambda)$,
where $A(m,\lambda) = C\lambda
    + \frac{C \lambda^2}{m^2} \exp \!\left( \frac{C \lambda}{m^2} \right)$.
\end{theorem}
\begin{proof}
This statement has essentially been proved
in~\cite{barashkov_stochastic_2022,gubinelli_fbsde_2026,bauerschmidt_logsobolev_2021,bauerschmidt_maximum_2022}.
We reprove it here in the context of lattice approximations and uniformly in the volume,
while also tracking the dependence on~$\lambda$ and~$m$.
We will opt for the method of Gubinelli and Meyer~\cite{gubinelli_fbsde_2026}.

By~\cite[Lemma~2.5]{gubinelli_fbsde_2026},
there exists an optimal~$h\in \mathbb{H}_a$ in Theorem~\ref{thm:boue-dupuis}
that satisfies almost surely the equation
\begin{equation}\label{eq:FBDSE}
h_{t}
=- \lambda\beta J_{t} \E [ \wick{\sin (\beta W_{\infty} + \beta I_{\infty}[h])} \mid \mathcal{F}_t ] .
\end{equation}
One can show \cite{barashkov_variational_2023,gubinelli_fbsde_2026}
that then $\mu_{\beta} = \operatorname{Law}(W+I)$.
Note that in this case we do not need to pass to the space $\overline{\mathcal{X}}$
to ensure existence of the minimizer.

Denoting $X_t = W_t + I_t$, we can rewrite~\eqref{eq:FBDSE} as a system.
For any $F\in C^2(\R_+ \times \R^\Lambda, \R^\Lambda)$,
we set $R_t = \lambda\beta \E [ \wick{\sin (\beta W_{\infty} + \beta I_{\infty}[h])} \mid \mathcal{F}_t ]
    - F_t (W_t+I_t(h))$.
As discussed in~\cite[Section~2.3]{gubinelli_fbsde_2026},
$R_t$ then solves the system
\begin{equation}
\begin{split}
R_t &= \E \left[ \int^{\infty}_t H_s (X_s) + \mathrm{D} F_s (X_s) J^2_s R_s \diff s
    \,\middle|\, \mathcal{F}_t \right]\!,\\
H_t &\coloneqq \partial_t F_t + \frac{1}{2} \operatorname{Tr} (J^2_t \mathrm{D}^2 F_t)
    -\frac{1}{2} \mathrm{D} \langle F_t, J^2_t F_t \rangle_{L^2 (\Lambda)} .
\end{split}
\end{equation}

In our setting we will set
\begin{equation}
F_t (\varphi) = \lambda \beta e^{\frac{\beta^2 c_{t}}{2}} \sin (\beta \varphi),
\end{equation}
where $c_t = \E [W_{t}(0)^2] = \log(t \wedge \delta^2)/4\pi + O(1)$.
(This estimate follows by adapting~\cite[Appendix~A]{gubinelli_fbsde_2026} to the periodic setup.)
It can be checked that this satisfies
\begin{equation}
\partial_t F_t + \frac{1}{2} \text{Tr} (J^2_t \mathrm{D}^2 F_t) = 0, \qquad
F_{\infty} (\varphi) = \lambda \beta \wick{\sin (\beta \varphi)}.
\end{equation}
In addition we have, with the notation $\lambda_{t}=\lambda e^{\frac{\beta^2 c_{t}}{2}}$, the identity
\begin{equation}
\langle F_t, J^2_t F_t \rangle_{L^2 (\Lambda)}
= \lambda^2_t \beta^2 \iint \sin (\beta \varphi) (x) C^{\ast 2}_t (x - y) \sin (\beta \varphi) (y) \dx\dy,
\end{equation}
where $C^{\ast 2}_t$ is as in Lemma~\ref{thm:heat kernel bounds}.
(For brevity we replace the sums over $\Lambda$ with integrals in convolutions.)
So taking the gradient we have
\begin{equation}
H_t = -\frac 1 2 \mathrm{D} \langle F_t, J^2_t F_t \rangle_{L^2 (\Lambda)}
= -\lambda^2_t \beta^3
    \iint \cos (\beta \varphi) (x) C^{\ast 2}_t (x - y) \sin (\beta \varphi) (y) \dx\dy.
\end{equation}
Observe that this implies the estimate
\begin{equation} \label{eq:estimate-bilinear}
\begin{split}
\norm{\mathrm{D} \langle F_t, J^2_t F_t \rangle_{L^2 (\Lambda)}}_{L^{\infty}(\Lambda)}
&\leq \lambda^2_t \beta^3 \norm{C^{\ast 2}_t}_{L^1}\\
&\leq \frac{1}{m^2} \beta^3 \lambda^2_t t^{-2}\\
&\leq \frac{1}{m^2} \lambda^2 \beta^3 \inorm{t}^{-1 - \varepsilon},
\end{split}
\end{equation}
where $\varepsilon$ is taken small enough.
In particular $R$ satisfies the bound
\begin{equation}
\begin{split}
\norm{R_t}_{L^{\infty}(\Prob, L^{\infty} (\Lambda))}
&\leq \frac{C \lambda^2}{m^2} \beta^3 + \lambda \beta^2
    \int_t^\infty \langle s \rangle^{\beta^2 / 4 \pi}
        \norm{J^2_s R_s}_{L^{\infty}(\Prob, L^{\infty} (\Lambda))} \diff s\\
&\leq \frac{C \lambda^2}{m^2} \beta^3 + \frac{\lambda \beta^2}{m^2}
    \int_t^\infty \langle s \rangle^{- 2 - \varepsilon + \beta^2 / 4 \pi}
        \norm{R_s}_{L^{\infty}(\Prob, L^{\infty} (\Lambda))} \diff s.
\end{split}
\end{equation}
From Grönwall's lemma we conclude
\begin{equation}
\norm{R_t}_{L^{\infty} (\Prob, L^{\infty} (\Lambda))}
\leq \frac{C \lambda^2}{m^2} \beta^3 \exp \left( \frac{C \lambda \beta^2}{m^2} \right)
\eqqcolon \tilde{A}(m,\lambda,\beta).
\end{equation}
This, \eqref{eq:estimate-bilinear}, and the heat kernel bounds
together imply that
\begin{equation}
\begin{split}
\norm{I_{\infty}(h)}_{\Holder^1}
&= \bignorm{\int_0^\infty  J^2_t (F_t(W_t+I_t(h))+R_t) \diff t}_{\Holder^1}\\
&\leq \sup_{t\geq 0} \left( \inorm{t}^{1/2+\varepsilon}
    \norm{J_t (F_t(W_t+I_t(h))+R_t)}_{L^\infty} \right)\\
&\leq C \sup_{t\geq 0}\left( \lambda \beta \inorm{t}^{1/2+\varepsilon+ \beta^2/8\pi-1}
    + \inorm{t}^{-1/2+\varepsilon} \norm{R_{t}}_{L^{\infty}}\right)\\
&\leq  C \lambda \beta + \tilde{A}(m,\lambda,\beta).
\end{split}
\end{equation}
We can dissolve $\beta < 4\pi$ into the implicit constant,
and this expression then gives our $A(m,\lambda)$.
\end{proof}

\begin{lemma}[Estimate on the Hessian]\label{lemma:bottom-product}
Let $W$ be a Gaussian free field
and $I \in L^{\infty} (\Prob, \Holder^{1 - \varepsilon}(\Lambda))$
such that $\norm{I}_{L^{\infty} (\Prob,\, \Holder^{1 - \varepsilon})} \leq  A(m,\lambda)$.
Then 
\[
\begin{split}
&\mathrel{\phantom{=}} \inf_{\norm{\psi}_{L^2} = 1}
    -\lambda \beta^2 \delta^2 \sum_{x \in \Lambda} \chi \wick{\cos(\beta W + \beta I)} \psi^2
    + m^2 \delta^2 \sum_{x \in \Lambda} \psi^2
    + \frac{\delta^2}{2} \sum_{x \in \Lambda} \abs{\nabla_{\delta} \psi}^2\\
&\geq -  \frac{\lambda A(m,\lambda)^{2-\varepsilon'}}{m^{2-\varepsilon'}}
   \left[ \norm{\wick{\sin (\beta W)}}^{2 - \varepsilon'}_{\Holder^{- \beta^2 / 4 \pi - \varepsilon} (\rho)}+\norm{\wick{\cos (\beta W)}}^{2 - \varepsilon'}_{\Holder_{\delta}^{- \beta^2 / 4 \pi - \varepsilon} (\rho)}\right]
\end{split}
\]
for $\varepsilon' > 0$ small enough.
\end{lemma}
\begin{proof}
We may write
\begin{equation}
\begin{split}
&\mathrel{\phantom{=}} \delta^2 \sum_{x \in \Lambda} \chi  \wick{\cos (\beta W + \beta I)} \psi^2\\
&= -\delta^2 \sum_{x \in \Lambda} \chi \wick{\sin (\beta W)} \sin (\beta I) \psi^2
    + \delta^2 \sum_{x \in \Lambda} \chi \wick{\cos (\beta W)} \cos (\beta I) \psi^2.
\end{split}
\end{equation}
We will now estimate the first term.
The second term is analogous.
We decompose the term as two paraproducts
\begin{equation}\label{eq:bottom-product paraproducts}
\begin{split}
&\mathrel{\phantom{=}} \delta^2 \sum_{x \in \Lambda} \chi \wick{\sin (\beta W)} \sin (\beta I) \psi^2\\
&= \delta^2 \sum_{x \in \Lambda} \chi [\psi \succeq \wick{\sin (\beta W)}] \sin (\beta I) \psi
    + \delta^2 \sum_{x \in \Lambda} \chi [\psi \prec \wick{\sin (\beta W)}] \sin (\beta I) \psi.
\end{split}
\end{equation}
Since we have assumed $\beta^2 / 4 \pi < 1$, we can estimate for $\varepsilon$ small that
\begin{equation}
\begin{split}
&\mathrel{\phantom{=}} \delta^2 \sum_{x \in \Lambda} \chi [\psi \succeq \wick{\sin (\beta W)}] \sin (\beta I) \psi\\
&\leq \norm{\psi}_{H^{1 - \varepsilon}(\Lambda)}
    \norm{\wick{\sin (\beta W)}}_{\Holder^{- \beta^2 / 4 \pi - \varepsilon} (\rho)}
    \norm{\sin (\beta I) \psi}_{L^2(\Lambda)}\\
&\leq \norm{\psi}_{H^{1 - \varepsilon}(\Lambda)}
    \norm{\wick{\sin (\beta W)}}_{\Holder^{- \beta^2 / 4 \pi - \varepsilon} (\rho)}.
\end{split}
\end{equation}
Now estimating the second term of~\eqref{eq:bottom-product paraproducts},
\begin{equation}
\begin{split}
&\mathrel{\phantom{=}} \delta^2 \sum_{x \in \Lambda} \chi (\psi \prec \wick{\sin (\beta W)}) \sin (\beta I) \psi\\
&\leq \norm{\psi}_{L^2(\Lambda)}
    \norm{\wick{\sin (\beta W)}}_{\Holder^{- \beta^2 / 4 \pi - \varepsilon} (\rho)}
    \norm{\sin(\beta I) \psi}_{H^{1 - \varepsilon}(\Lambda)}\\
&\lesssim \norm{\wick{\sin (\beta W)}}_{\Holder^{-\beta^2 / 4 \pi - \varepsilon} (\rho)}
    \norm{\psi}_{H^{1 - \varepsilon}(\Lambda)}
    \norm{I}_{\Holder^{1 - \varepsilon}(\Lambda)}\\
&\lesssim A(m,\lambda) \norm{\wick{\sin (\beta W)}}_{\Holder^{- \beta^2 / 4 \pi - \varepsilon} (\rho)}
    \norm{\psi}_{H^{1 - \varepsilon}(\Lambda)},
\end{split}
\end{equation}
where the estimate
$\norm{\sin(\beta I)}_{\Holder^{1-\varepsilon}} \lesssim \norm{I}_{\Holder^{1-\varepsilon}}$
follows by Lipschitz continuity.
Collecting the terms we obtain the bound
\begin{equation}
\lambda \beta^2 \bigabs{ \delta^2 \sum_{x \in \Lambda} \chi \psi^2 \wick{\sin (\beta W)} \sin (\beta I) }
\leq \lambda C A(m,\lambda)
    \norm{\wick{\sin (\beta W)}}_{\Holder^{- \beta^2 / 4 \pi - \varepsilon} (\rho)}
    \norm{\psi}^{1 - \varepsilon}_{H^1}.
\end{equation}
Thus by Young's inequality we get 
\begin{equation}
\begin{split}
&\mathrel{\phantom{=}} \lambda \beta^2 \delta^2 \sum_{x \in \Lambda} \chi \psi^2 \wick{\cos (\beta W + \beta I)}
    + m^2 \delta^2 \sum_{x \in \Lambda} \psi^2
    + \frac{\delta^2}{2} \sum_{x \in \Lambda} \abs{\nabla_{\delta} \psi}^2\\
&\gtrsim - \frac{\lambda A(m,\lambda)^{2-\varepsilon'}}{m^{2-\varepsilon'}} \left(
    \norm{\wick{\sin (\beta W)}}^{2 - \varepsilon'}_{\Holder^{- \beta^2 / 4 \pi - \varepsilon}(\rho)}
    + \norm{\wick{\cos (\beta W)}}^{2 - \varepsilon'}_{\Holder^{- \beta^2 / 4 \pi - \varepsilon}(\rho)} \right).
\end{split}
\end{equation}
\end{proof}

To conclude it remains to show that $(\lambda/m)^2
\norm{\wick{\sin (\beta W)}}^{2 - \varepsilon'}_{\Holder^{- \beta^2 / 4 \pi - \varepsilon} (\rho)}$,
as well as the analogous cosine,
are exponentially integrable.
For this we will need the following lemma.

\begin{lemma}
\label{lem:estimate-sin-product}
For $\beta^2 < 4\pi$, we can bound
\[
\begin{split}
&\mathrel{\phantom{=}}
\norm{\wick{\cos (\beta W + \beta I)}}_{\Holder^{- \beta^2 / 4\pi - \varepsilon} (\rho)}
+ \norm{\wick{\sin (\beta W + \beta I)}}_{\Holder^{- \beta^2 / 4\pi - \varepsilon} (\rho)}\\
&\lesssim \left(
    \norm{\wick{\cos (\beta W)}}_{\Holder^{- \beta^2 / 4\pi - \varepsilon}(\rho)}
    + \norm{\wick{\sin (\beta W)}}_{\Holder^{- \beta^2 / 4\pi - \varepsilon}(\rho)}
    \right) (1 + \norm{I}_{H^1(\Tor^2)}).
\end{split}
\]
\end{lemma}
\begin{proof}
By the usual trigonometric identity and Definition~\ref{def:wick},
\begin{equation}
\wick{\cos (\beta W + \beta I)}
= \wick{\cos (\beta W)} \cos (\beta I) - \wick{\sin (\beta W)} \sin (\beta I).
\end{equation}
We again use a paraproduct decomposition (Theorem~\ref{thm:besov paraproducts}).
First,
\begin{equation}
\norm{\wick{\cos(\beta W)} \prec \cos(\beta I)}_{\Holder^{-\beta^2/4\pi - \epsilon}(\rho)}
\lesssim 
    \norm{\wick{\cos(\beta W)}}_{\Holder^{-\beta^2/4\pi - \epsilon}(\rho)}
    \norm{\cos(\beta I)}_{L^\infty(\Tor^2)},
\end{equation}
where the second factor is trivially bounded.
The same estimate holds for the opposite paraproduct as well.
The resonant term is then estimated with
\begin{equation}
\begin{split}
\norm{\wick{\cos(\beta W)} \odot \cos(\beta I)}_{\Holder^{-\beta^2/4\pi - \epsilon}(\rho)}
&\lesssim \norm{\wick{\cos(\beta W)} \odot \cos(\beta I)}_{H^{1 - \beta^2/4\pi - \epsilon}(\rho)}\\
&\lesssim \norm{\wick{\cos(\beta W)}}_{\Holder^{- \beta^2/4\pi - \epsilon}(\rho)}
    \norm{\cos(\beta I)}_{H^1(\Tor^2)}
\end{split}
\end{equation}
and the chain rule.
The estimate on the sine is analogous. 
\end{proof}

\begin{lemma} \label{lem:estimate-sin-moment}
There exists $p<\infty$ such that for any $q>0$ the following estimate holds:
\begin{equation*}
\log \E \exp \left(
    q \norm{\wick{\sin (\beta W)}}^{2 - \varepsilon'}_{\Holder^{-\beta^2/4\pi-\varepsilon}(\rho)}
  + q \norm{\wick{\cos (\beta W)}}^{2 - \varepsilon'}_{\Holder^{-\beta^2/4\pi-\varepsilon}(\rho)} \right)
\lesssim \left(\frac{q}{(m\wedge1)^{2+\beta^2/2\pi}}\right)^p.
\end{equation*}
\end{lemma}
\begin{proof}
To show exponential tails we apply the Bou{\'e}--Dupuis formula to obtain
\begin{equation}\label{eq:estimate-sin-moment bd}
\begin{split}
&\mathrel{\phantom{=}} \log \E \exp \left(
    q \norm{\wick{\sin (\beta W)}}^{2 - \varepsilon'}_{\Holder^{-\beta^2/4\pi-\varepsilon}(\rho)}
  + q \norm{\wick{\cos (\beta W)}}^{2 - \varepsilon'}_{\Holder^{-\beta^2/4\pi-\varepsilon}(\rho)} \right)\\
&= \inf_{h \in \mathbb{H}_a} \E \left[
    q \norm{\wick{\sin (\beta W + \beta I[h])}}^{2-\varepsilon'}_{\Holder^{-\beta^2/4\pi-\varepsilon}(\rho)}
    + q \norm{\wick{\cos (\beta W+\beta I[h])}}^{2 - \varepsilon'}_{\Holder^{-\beta^2/4\pi-\varepsilon}(\rho)}
    + \frac{1}{2} \int^{\infty}_0 \!\! \norm{h_t}^2_{L^2} \diff t \right]\!.
\end{split}
\end{equation}
We now apply Lemma~\ref{lem:estimate-sin-product} to estimate the expectation from below by
\begin{equation}
-C q \! \left[
    \norm{\wick{\cos (\beta W)}}_{\Holder^{- \beta^2 / 4\pi - \varepsilon}(\rho)}
    + \norm{\wick{\sin (\beta W)}}_{\Holder^{- \beta^2 / 4\pi - \varepsilon}(\rho)}
\right]^{2 - \varepsilon'} \!\! (1+ \norm{I[h]}_{H^1})^{2 - \varepsilon'}
+ \frac{1}{2} \int^{\infty}_0 \!\! \norm{h_t}^2_{L^2} \diff t.
\end{equation}
The definition of $J_s$ and Lemma~\ref{thm:besov hs norm} imply
\begin{equation}
\begin{split}
(m^2 \wedge 1) \norm{I[h]}_{H^1}^2
&\simeq \sum_{k \in \hat\Lambda} (m^2 \wedge 1) (1 + \abs k^2) \left[ \int_0^\infty
    \frac 1 t \exp\!\left( -\frac{m^2 + \abs k^2}{2t} \right) \hat h(t, k) \diff t \right]\\
&\leq \sum_{k \in \hat\Lambda} ((m^2 \wedge 1) + \abs k^2) \left[ \int_0^\infty
    \frac{1}{t^2} \exp\!\left( -\frac{m^2 + \abs k^2}{t} \right) \diff t \right]
    \!\left[ \int_0^\infty \hat h(t, k)^2 \diff t \right]\\
&\leq \int_0^\infty \norm{h_t}_{L^2}^2 \diff t.
\end{split}
\end{equation}
We therefore obtain by Young's inequality that
\begin{equation}
\begin{split}
\eqref{eq:estimate-sin-moment bd}
&\geq -C [q / (m^2 \wedge 1)]^{2/\varepsilon'} \E \left[ 1
        + \norm{\wick{\cos (\beta W)}}_{\Holder^{- \beta^2 / 4\pi - \varepsilon}(\rho)}
        + \norm{\wick{\sin (\beta W)}}_{\Holder^{- \beta^2 / 4\pi - \varepsilon}(\rho)}
    \right]^{4/\varepsilon'}\\
&\qquad{} - \frac 1 4 (m^2 \wedge 1) \norm{I[h]}_{H^1}^2
    + \frac 1 2 \int_0^\infty \norm{h_t}_{L^2}^2 \diff t\\
&\geq -C \left( \frac{q}{(m \wedge 1)^{2+\beta^2/2\pi}} \right)^{2/\varepsilon'},
\end{split}
\end{equation}
where we evaluated the expectation with Lemma~\ref{moment-gff-sine}.
\end{proof}

Thus we can conclude the following theorem.

\begin{theorem}\label{thm:sine-gordon finished}
There exists $p \in (1,\infty)$, independent of $m$, $\lambda$, and $\beta$,
such that $Q(u)$ defined by~\eqref{eq:Q} satisfies
\[
\log \int \exp\left(\frac{4}{\sigma}Q(u)\right) \diff\mu_{\beta}(u)
\lesssim \left(
    \frac{4\lambda A(m,\lambda)^{2-\varepsilon'}}
        {\sigma (m \wedge 1)^{4-\varepsilon'+\beta^2/2\pi}}
    \right)^p .
\]
\end{theorem}
\begin{proof}
By Theorem~\ref{prop:coupling} we may write
\begin{equation}
\int \exp \left( \frac{4}{\sigma} Q (u) \right) \diff \mu_\beta (u)
= \E \left[ \exp \left( \frac{4}{\sigma} Q (W + I) \right) \right],
\end{equation}
and by Lemma~\ref{lemma:bottom-product} we have
\begin{equation}
\abs{Q (W + I)}
\leq \frac{\lambda A(m,\lambda)^{2-\varepsilon'}}{(m\wedge1)^{2-\varepsilon'}} \left[
    \norm{\wick{\sin(\beta W)}} ^{2 - \varepsilon'}_{\Holder^{- \beta^2 / 4 \pi - \varepsilon} (\rho)}
    \norm{\wick{\cos(\beta W)}} ^{2 - \varepsilon'}_{\Holder^{- \beta^2 / 4 \pi - \varepsilon} (\rho)}
    \right],
\end{equation}
from which we can conclude by Lemma~\ref{lem:estimate-sin-moment}.
\end{proof}

The constant~$A$ is bounded by a constant when $\lambda \leq m^2$.
By~\cite{bauerschmidt_logsobolev_2021}, the log-Sobolev constant satisfies $\sigma \gtrsim m^2$.
Hence the right-hand side in Theorem~\ref{thm:sine-gordon finished} is bounded by a constant
when $\lambda \leq m^{6 + \beta^2/2\pi}$.
This finishes the proof of Theorem~\ref{thm:integrability}.

\section{The \texorpdfstring{$\Phi^4_3$}{Phi43} model}\label{sec:phi43}
In this section we prove Theorem~\ref{thm:3 dimensions}.
As in that theorem, we fix the domain size $L = 1$.

To lighten the notation we assume $0<\lambda< m^2\leq 1$.
One can obtain $R(q,\lambda,m)<\infty$ for general $\lambda$ and $m$ with minor modifications of the argument. 

\subsection{An SPDE method for estimating tails}

For the $\Phi_3^4$ measure, we use a different method to estimate the tail of
$Q$. Besides providing an alternative approach, this lets us reuse results on
stochastic quantization.
In particular, we rely on \cite{gubinelli_pde_2021}, which studies the lattice
approximation for which \cite{bauerschmidt_logsobolev_2024} establishes
logarithmic Sobolev inequalities.

Recall that the $\Phi_3^4$ measure on a lattice $\Lambda \subseteq \delta \Z^3 \cap \Tor^3$ is given by
\begin{equation}
\diff\mu = \frac{1}{Z} \exp \left(- \lambda \delta^3\sum_{x \in \Lambda} \wick{u^4}
    - \lambda^2 a_{\delta} \delta^3 \sum_{x \in \Lambda} \wick{u^2}
    - \frac{m^2}{2} \delta^3 \sum_{x \in \Lambda} u^2
    - \frac{1}{2} \delta^3\sum_{x \in \Lambda} \abs{\nabla_{\delta} u}^2 \right) \diff u,
\end{equation}
where $Z$ is a normalization factor chosen such that $\mu$ is a probability measure.
For $q > 0$, let $\nu^{q Q}$ be the measure defined by
\begin{equation}
\diff \nu^{q Q} = \frac{1}{Z^{qQ}} \exp (q Q (u)) \diff \mu,
\qquad Z^{q Q} = \int \exp (q Q (u)) \diff \mu.
\end{equation}
For fixed $\delta$, the function $Q$ is bounded above, so $Z^{qQ}<\infty$.
Our strategy is based on the following lemma.

\begin{lemma}\label{lem:hairer-steele}
\[ Z^{q Q} \leq \exp \!\left( \int q Q (u) \diff \nu^{q Q} (u) \right)\!. \]
\end{lemma}
\begin{proof}
Since $\diff\nu^{qQ}=e^{qQ}\diff\mu/Z^{qQ}$, we have
\[
\int e^{-qQ(u)}\diff\nu^{qQ}(u)=\frac{1}{Z^{qQ}}.
\]
Jensen's inequality therefore gives
\[
-\int qQ(u)\diff\nu^{qQ}(u)
\leq \log\int e^{-qQ(u)}\diff\nu^{qQ}(u)
=-\log Z^{qQ},
\]
which is the claimed bound.
\end{proof}

To estimate the right-hand side,
we realize $\nu^{qQ}$ as the stationary
measure of the following equation, where $\xi$ is space-time white noise on~$\Lambda$:
\begin{equation}\label{eq:sq-phi43}
(\partial_t + m^2 - \Laplace_\delta) u + 4 \lambda \wick{u^3}
    + 2 \lambda^2 a_{\delta} u - q \nabla Q (u)
= \sqrt{2} \xi.
\end{equation}
We will make use of the following
proposition from \cite{gubinelli_pde_2021}.
For $K \colon \Lambda \to \R$, we define
\begin{equation}
\mathcal{E} (K) = \lambda \delta^3 \sum_{x \in \Lambda} K^4 (x)
    + m^2 \norm{K}^2_{H^{1 - \varepsilon}(\Lambda)}.
\end{equation}

\begin{proposition}
\label{prop-gub-hov}
Let $u = W + W^{[3]} + K$, where $W$ and $W^{[3]}$ are the
stationary solutions to
\[
(\partial_t+m^2-\Laplace_\delta)W
    =\sqrt2 \, \proj_{\geq1}\xi, \qquad
(\partial_t+m^2-\Laplace_\delta)W^{[3]}
    =-4\lambda\wick{W^3}.
\]
and $K$ is a stationary process. Then
\[
\E \left[ \delta^3 \sum_{x \in \Lambda}
    \big((m^2-\Laplace_\delta)K
    +4\lambda\bigl(\wick{u^3}-\wick{W^3}\bigr)
    +2\lambda^2a_{\delta}u\big)K \right]
\geq \E \mathcal{E} (K) - \E \mathcal{B}^3 (W, \lambda, m).
\]
where $\mathcal{B}^3 (W, \lambda, m)$ is a stochastic term that satisfies
\[
\E \abs{\mathcal{B}^3 (W, \lambda, m)} \leq p (\lambda, m)
\]
where $p (\lambda, m)$ is a rational function of $\lambda$ and $m$ such that
$\abs{p (\lambda, m)} \leq C$ as long as $\lambda \leq m^2 \leq 1$.
\end{proposition}
\begin{proof}
This follows as Lemma~4.4 and Theorem~4.5 from~\cite{gubinelli_pde_2021},
with minor modifications to obtain the required behaviour in $\lambda$, $m$,
and setting the weight $\rho = 1$, as we are on a torus of fixed volume. 
\end{proof}

\begin{remark}
We have excluded the zero mode of~$\xi$ from~$W$ to make sure that
$m^2 - \Laplace_\delta$ has a spectral gap uniformly in~$m$.
This will make the bounds on stochastic terms built from $W$ independent of~$m$.
The price to pay is that we will have to include the corresponding zero mode term into~$K$,
however this will not present any significant problems. 
\end{remark}

\subsection{Analysis of the Hessian}

In this section we establish the necessary estimates on $\Hess U$ for the
$\phi_3^4$ model. Recall that in this case
\begin{equation}
\quad \Hess U(u) \varphi = 12 \lambda \wick{u^2} \varphi
    + 2 \lambda^2 a_{\delta} \varphi + m^2 \varphi - \Laplace_\delta \varphi.
\end{equation}
Write $\mathcal {A}_{\delta}(u)(\varphi)=\langle\varphi,\Hess U(u)\varphi\rangle_{L^2}$.
Recall that in three dimensions $\norm{\wick{u^2}}_{\Holder^{- \alpha}}$
is finite only for $\alpha > 1$. Thus we cannot
control $Q(u)$ by the method used in the preceding subsections.

The same regularity issue also occurs in the literature on
the Anderson Hamiltonian; see~\cite{allez_continuous_2015,gubinelli_semilinear_2020,labbe_continuous_2019}.
In fact the scaling we encounter here corresponds precisely to the two-dimensional Anderson Hamiltonian.
We therefore begin by carrying out a paracontrolled analysis of this singular Schr{\"o}dinger operator,
following the aforementioned works on the Anderson Hamiltonian.
This analysis summarised in Proposition~\ref{3d-form-wick-square} below is done only for the Gaussian part.
Then we will treat the non-Gaussian (but more regular) parts as perturbations. 

\begin{definition}
\label{rem:hessian-coefficients}
Set
\[
A_\delta^{\circ2} = (1-\Laplace_\delta)^{-1/2}\wick{W^2},
\qquad
W^{2\diamond2} = A_\delta^{\circ2}-\frac{a_\delta}{72}.
\]
\end{definition}

The coefficient of $W^{2 \diamond2}$ arises as follows:
A square-completion in Proposition~\ref{3d-form-wick-square} produces $-144\lambda^2A_\delta^{\circ2}$,
while the Hessian contributes $2\lambda^2a_\delta$.
Their divergent parts therefore combine as
$-144\lambda^2(A_\delta^{\circ2}-a_\delta/72)$,
and $W^{2 \diamond 2}$ will turn out to be in $\Holder^{-\varepsilon}$ uniformly in $\delta$.

\begin{proposition}[Hessian estimate for Gaussian argument]
\label{3d-form-wick-square}
For every $\varepsilon>0$ small enough, there exists a family of random variables $\Xi_{\delta}>0$ such that:
\begin{itemize}
\item for every $p<\infty$, we have $ \sup_{\delta>0}\E[\Xi_{\delta}^p]\leq C_p$; and
\item for every $\varphi$ with $\norm{\varphi}_{L^2(\Lambda)} = 1$, we have
    \[
    \mathcal {A}_\delta(W)(\varphi)
    \geq (1-\lambda/2)
        \norm{\varphi}_{H^{1-\varepsilon}(\Lambda)}^2
        -C_\varepsilon \lambda^{1-2\varepsilon}\Xi_{\delta}-1.
    \]
\end{itemize}
\end{proposition}
\begin{proof}
Set $h = (1 - \Laplace_\delta)^{1 / 2} \varphi$.
Leaving the renormalization term $2\lambda^2a_{\delta} \varphi^2$ in~$\mathcal A_\delta(W)(\varphi)$
aside for the moment,
and utilizing $\norm{\varphi}_{L^2 (\Lambda)} = 1$, we are left with
\begin{equation}\label{eq:3d-form-wick-square simple hessian}
12\lambda\delta^3 \sum_{x \in \Lambda} \wick{W^2} \varphi^2
    + \delta^3 \sum_{x \in \Lambda} h^2
    - (1 - m^2).
\end{equation}
The last term is bounded below by $-1$ uniformly in $m^2>0$.
We can rewrite the first term as
\begin{equation}\label{eq:3d-form-wick square first}
\begin{split}
&\mathrel{\phantom{=}} 12\lambda \langle \wick{W^2} \varphi, \varphi \rangle_{L^2}\\
&= 12\lambda \langle \wick{W^2} \succ \varphi, \varphi \rangle_{L^2}
    + 12\lambda \langle \wick{W^2} \circ \varphi, \varphi \rangle_{L^2}
    + 12\lambda \langle \wick{W^2} \prec \varphi, \varphi \rangle_{L^2}\\
&= 24\lambda \langle \wick{W^2} \succ \varphi, \varphi \rangle_{L^2}
    + 12\lambda \mathfrak{C}_1 (W, \varphi),
\end{split}
\end{equation}
where
\begin{equation}
\mathfrak{C}_1 (W, \varphi) \coloneqq
    \langle \wick{W^2} \prec \varphi, \varphi \rangle_{L^2}
    + \langle \wick{W^2} \circ \varphi, \varphi \rangle_{L^2}
    - \langle \wick{W^2} \succ \varphi, \varphi \rangle_{L^2}.
\end{equation}
The commutator $\mathfrak{C}_1 (\wick{W^2}, \varphi)$ is estimated in
Lemma~\ref{first-commutator} below.
We then remove the first term in~\eqref{eq:3d-form-wick square first}
with an ansatz for $h$.
We observe that
\begin{equation}
\langle \wick{W^2} \succ \varphi, \varphi \rangle_{L^2}
= \langle (1 - \Laplace_\delta)^{- 1 / 2} (\wick{W^2} \succ \varphi), h \rangle_{L^2}.
\end{equation}
Set
\begin{equation}
Z_\varphi \coloneqq (1-\Laplace_\delta)^{-1/2}(\wick{W^2}\succ\varphi),
\end{equation}
and define $w$ by $h=-12\lambda Z_\varphi+w$.
Then
\begin{equation}
\begin{split}
\norm{h}_{L^2}^2
&= -24\lambda\langle Z_\varphi,w\rangle_{L^2}
    + 144\lambda^2 \norm{Z_\varphi}_{L^2}^2 + \norm{w}_{L^2}^2\\
&= -24\lambda\langle Z_\varphi,h\rangle_{L^2}
    - 144\lambda^2 \norm{Z_\varphi}_{L^2}^2 + \norm{w}_{L^2}^2.
\end{split}
\end{equation}
Thus the first two terms in~\eqref{eq:3d-form-wick-square simple hessian} equal
\begin{equation}
24\lambda\langle\wick{W^2}\succ\varphi,\varphi\rangle_{L^2}
    + \norm{h}_{L^2}^2
= -144\lambda^2 \norm{Z_\varphi}_{L^2}^2
    + \norm{w}_{L^2}^2
    + 12\lambda \mathfrak{C}_1 (W, \varphi).
\end{equation}

We next analyze $\norm{Z_\varphi}_{L^2}^2$.
Recall the notation
$A_\delta^{\circ2}=(1-\Laplace_\delta)^{-1/2}\wick{W^2}$ and set
\begin{equation}
\mathfrak C_2(\wick{W^2},\varphi)
= Z_\varphi - A_\delta^{\circ2} \succ \varphi.
\end{equation}
The commutator estimate \cite[Proposition~A.2]{allez_continuous_2015}
gives for $\alpha \in [0,1-\varepsilon)$ that
\begin{equation}
\norm{\mathfrak C_2(\wick{W^2},\varphi)}_{H^{\alpha}}
\lesssim \norm{\wick{W^2}}_{\Holder^{-1-\varepsilon}} \norm{\varphi}_{H^{\alpha+\varepsilon}}.
\end{equation}
Hence
\begin{equation}\label{eq:3d-form-wick-square a norm}
\norm{Z_\varphi}_{L^2}^2
= \norm{A_\delta^{\circ2} \succ\varphi}_{L^2}^2
    + \norm{\mathfrak C_2(\wick{W^2},\varphi)}_{L^2}^2
    + 2\langle A_\delta^{\circ2} \succ\varphi, \mathfrak C_2(\wick{W^2},\varphi) \rangle_{L^2},
\end{equation}
and using the estimate for $\mathfrak{C}_2$ with Besov interpolation,
\begin{equation}
\begin{split}
&\mathrel{\phantom{=}} \norm{\mathfrak C_2(\wick{W^2},\varphi)}_{L^2}^2
    + 2\langle A_\delta^{\circ2} \succ\varphi,  \mathfrak C_2(\wick{W^2},\varphi) \rangle_{L^2}\\
&\lesssim \norm{\wick{W^2}}_{\Holder^{-1-\varepsilon}}^2 \norm{\varphi}_{H^\varepsilon}^2
    + \norm{(1-\Laplace_\delta)^{-1/2}\wick{W^2} \succ \varphi}_{H^{-2\varepsilon}}
        \norm{\mathfrak C_2(\wick{W^2},\varphi)}_{H^{2\varepsilon}}\\
&\lesssim \norm{\wick{W^2}}_{\Holder^{-1-\varepsilon}}^2
        \norm{\varphi}_{H^{1-\varepsilon}}^{2\varepsilon/(1-\varepsilon)}.
\end{split}
\end{equation}
To deal with the first term of~\eqref{eq:3d-form-wick-square a norm}
we observe that by Lemma~\ref{c3} below,
\begin{equation}
\mathfrak{C}_3(\wick{W^2}, \varphi) \coloneqq
\bignorm{((1 - \Laplace_\delta)^{- 1 / 2} \wick{W^2} \succ \varphi)^2
    - ((1 - \Laplace_\delta)^{- 1 / 2} \wick{W^2})^2 \varphi^2}_{L^1}
\end{equation}
is bounded by $C \norm{\wick{W^2}}^2_{\Holder^{-1-\varepsilon}}
    \norm{\varphi}_{H^{1-\varepsilon}}^{3\varepsilon/(1-\varepsilon)}$.

At this point, we have bounded~\eqref{eq:3d-form-wick square first} from below by
\begin{equation}
-144\lambda^2 \norm{((1 - \Laplace_\delta)^{- 1 / 2} \wick{W^2}) \varphi}_{L^2}^2
    + \norm{w}_{L^2}^2
    - 1
    - C(\lambda + \lambda^2) (1 + \norm{\wick{W^2}}_{\Holder^{- 1 - \varepsilon}}^2) \Bigl[ 
        1 + \norm{\varphi}^{1 + 4 \varepsilon}_{H^{1 - \varepsilon}} \Bigr].
\end{equation}
Let us now add back the renormalization term $2\lambda^2 a_\delta \varphi^2$
to cancel the divergence of the first term.
They combine as
\begin{equation}
-144\lambda^2 \delta^3 \sum_{x \in \Lambda}
    \left(((1 - \Laplace_\delta)^{-1/2}\wick{W^2})^{2} -\frac{a_{\delta}}{72}\right) \varphi^2.
\end{equation}
With the notation of Definition~\ref{rem:hessian-coefficients}, we can estimate
\begin{equation}
\lambda^2 \bigabs{\delta^3 \sum_{x \in \Lambda}
    \left(A_\delta^{\circ2} -\frac{a_{\delta}}{72}\right) \varphi^2}
\lesssim \lambda^2 \norm{W^{2\diamond2}}_{\Holder^{-\varepsilon}} \norm{\varphi}_{H^{\varepsilon}}^2
\lesssim \lambda^2 \norm{W^{2\diamond2}}_{\Holder^{-\varepsilon}} \norm{\varphi}_{H^{1-\varepsilon}}^{3\varepsilon}.
\end{equation}
Therefore we have showed
\begin{equation}\label{eq:anderson-proper}
\mathcal{A}_{\delta}(W)(\varphi) \geq \norm{w}_{L^2}^2 + R(W, \varphi) - 1,
\end{equation}
where the remainder satisfies  
\begin{equation}
\abs{R (W, \varphi)}
\leq C (\lambda + \lambda^2) (1 + \norm{W^{2 \diamond 2}}_{\Holder^{- \varepsilon}}
        + \norm{\wick{W^2}}^2_{\Holder^{- 1 - \varepsilon}})
    (1 + \norm{\varphi}^{1+4\varepsilon}_{{H^{1 - \varepsilon}} }).
\end{equation}
  
Let us abbreviate
\begin{equation}
X \coloneqq 1 + \norm{W^{2\diamond2}}_{\Holder^{-\varepsilon}}
    + \norm{\wick{W^2}}_{\Holder^{-1-\varepsilon}}^2.
\end{equation}
We also use the definitions of~$h$ and~$w$, paraproduct estimates,
and the assumption $\norm{\varphi}_{L^2} = 1$, to observe
\begin{equation}
\begin{split}
\norm{\varphi}_{{H^{1 - \varepsilon}} }
&\leq 12\lambda \norm{(1 - \Laplace_\delta)^{- 1 / 2} (\wick{W^2} \succ \varphi)}_{H^{- \varepsilon}}
    + \norm{w}_{L^2}\\
&\leq C \lambda \norm{\wick{W^2}}_{\Holder^{- 1 - \varepsilon}}
    \norm{\varphi}_{L^2}
    + \norm{w}_{L^2}\\
&\leq C \lambda \norm{\wick{W^2}}_{\Holder^{- 1 - \varepsilon}}
    + \norm{w}_{L^2}.
\end{split}
\end{equation}
Therefore the remainder can be further estimated as
\begin{equation}
R(W,\varphi)+\norm{w}_{L^2}^2
\geq \norm{w}_{L^2}^2
    -C(\lambda+\lambda^2)X
    (1 + \norm{\wick{W^2}}_{\Holder^{-1-\varepsilon}}^{1+4\varepsilon}
        + \norm{w}_{L^2}^{1+4\varepsilon}).
\end{equation}
As we assumed $\lambda \leq 1$, the final term can be estimated with Young's inequality as
\begin{equation}
-C_\varepsilon \lambda \left[
    X (1 + \norm{\wick{W^2}}_{\Holder^{-1-\varepsilon}}^{1+4\varepsilon})
    + X^{r_\varepsilon} \right]
- \frac \lambda 2 \norm{w}_{L^2}^2,
\quad\text{where } r_\varepsilon = \frac{2}{1 - 4\varepsilon}.
\end{equation}
The stochastic quantities $X$ and $\norm{\wick{W^2}}_{\Holder^{-1-\varepsilon}}$
have moments of all orders; see~\cite[Section~4.1]{gubinelli_pde_2021}
for the corresponding lattice stochastic estimates.
Hence we can replace the stochastic terms with a random variable $\Xi$
that has uniform-in-$\delta$ bounds on moments.

Finally, by definition
\begin{equation}
\norm{w}_{L^2}^2
\geq \norm{\varphi}_{H^{1-\varepsilon}}^2
    -144\lambda^2\norm{\wick{W^2}}_{\Holder^{-1-\varepsilon}}^2.
\end{equation}
Absorbing the last stochastic term into $\Xi$ proves the estimate.
\end{proof}

\begin{lemma}\label{first-commutator}
Let
\[
\mathfrak{C}_1 (\wick{W^2}, \varphi)
\coloneqq \langle \wick{W^2} \prec \varphi, \varphi \rangle_{L^2 (\Lambda)}
    + \langle \wick{W^2} \circ \varphi, \varphi \rangle_{L^2 (\Lambda)}
    - \langle \wick{W^2} \succ \varphi, \varphi \rangle_{L^2 (\Lambda)}.
\]
The following estimate holds for $\varphi \in H^{1 - \varepsilon}$
satisfying $\norm{\varphi}_{L^2} \leq 1$:
\[
\abs{\mathfrak{C}_1 (\wick{W^2}, \varphi)}
\lesssim (1 + \norm{\wick{W^2}}_{\Holder^{- 1 - \varepsilon}})
    \norm{\varphi}^{1 + 4 \varepsilon}_{H^{1 - \varepsilon}}.
\]
\end{lemma}
\begin{proof}
In the first term we have by duality,
the paraproduct estimate (Lemma~\ref{thm:besov paraproducts}),
and interpolation
\begin{equation}
\langle \wick{W^2} \prec \varphi, \varphi \rangle_{L^2 (\Lambda)}
\lesssim \norm{\wick{W^2}}_{\Holder^{- 1 - \varepsilon}}
    \norm{\varphi}_{H^{1 - \varepsilon}}
    \norm{\varphi}_{H^{3\varepsilon}}
\lesssim (1 + \norm{\wick{W^2}}_{\Holder^{- 1 - \varepsilon}})
    \norm{\varphi}^{1 + 4 \varepsilon}_{H^{1 - \varepsilon}}.
\end{equation}
For the second term we apply \cite[Lemma~A.6]{gubinelli_semilinear_2020} to get
\begin{equation}
\begin{split}
\bigabs{\langle\wick{W^2}\circ\varphi,\varphi\rangle_{L^2}
    - \langle\wick{W^2}\succ\varphi,\varphi\rangle_{L^2}}
&\lesssim
    \norm{\wick{W^2}}_{\Holder^{-1-\varepsilon}}
    \norm{\varphi}_{H^{1-\varepsilon}}
    \norm{\varphi}_{H^{3\varepsilon}}\\
&\lesssim
    \left(1+\norm{\wick{W^2}}_{\Holder^{-1-\varepsilon}}\right)
    \norm{\varphi}_{H^{1-\varepsilon}}^{1+4\varepsilon}.
\end{split}
\end{equation}
\end{proof}

\begin{lemma}\label{c3}
The commutator term
\[
\mathfrak{C}_3 (\wick{W^2}, \varphi)
\coloneqq \delta^3 \sum_{x \in \Lambda} \Bigl[ ((1 - \Laplace_\delta)^{- 1 / 2} \wick{W^2}
    \succ \varphi)^2 - ((1 - \Laplace_\delta)^{- 1 / 2} \wick{W^2})^2 \varphi^2 \Bigr]
\]
satisfies
\[
\abs{\mathfrak{C}_3 (\wick{W^2}, \varphi)}
\lesssim \norm{\wick{W^2}}^2_{\Holder^{-1-\varepsilon}}
     \norm{\varphi}_{H^{3\varepsilon}} \norm{\varphi}_{L^2}.
\]
\end{lemma}
\begin{proof}
With $A_\delta^{\circ2}=(1-\Laplace_\delta)^{-1/2}\wick{W^2}$, we note that
\begin{equation}
(A_\delta^{\circ2})^2\varphi^2 - (A_\delta^{\circ2} \succ \varphi)^2
= (A_\delta^{\circ2} \preccurlyeq \varphi)^2
    + 2(A_\delta^{\circ2} \succ \varphi)(A_\delta^{\circ2} \preccurlyeq \varphi).
\end{equation}
Now the second term can be estimated by
\begin{equation}
\begin{split}
&\mathrel{\phantom{=}} \bigabs{\delta^3 \sum_{x \in \Lambda}
    ((1 - \Laplace_\delta)^{- 1 / 2} \wick{W^2} \succ \varphi)
    ((1 - \Laplace_\delta)^{- 1 / 2} \wick{W^2} \preccurlyeq \varphi)}\\
&\lesssim \norm{(1 - \Laplace_\delta)^{- 1 / 2} \wick{W^2} \succ \varphi}_{H^{-2\varepsilon}}
    \norm{(1 - \Laplace_\delta)^{- 1 / 2} \wick{W^2} \preccurlyeq \varphi}_{H^{2\varepsilon}}\\
&\lesssim \norm{(1 - \Laplace_\delta)^{- 1 / 2} \wick{W^2}}_{\Holder^{- \varepsilon}}
    \norm{\varphi}_{L^2}
    \norm{(1 -\Laplace_\delta)^{- 1 / 2} \wick{W^2}}_{\Holder^{-\varepsilon}}
    \norm{\varphi}_{H^{3 \varepsilon}}\\
&\lesssim \norm{\wick{W^2}}^2_{\Holder^{- 1 - \varepsilon}}
    \norm{\varphi}_{H^{3 \varepsilon}}
    \norm{\varphi}_{L^2}.
\end{split}
\end{equation}
The first term can be bounded analogously, and we can conclude. 
\end{proof}

We now turn towards deducing a lower bound for $\mathcal {A}$ with an argument
that is a Gaussian combined with a more regular perturbation. 

\begin{lemma}[Hessian estimate for $\Phi^4_3$]
\label{lem:lower-bound}
With $\Xi$ as in Proposition~\ref{3d-form-wick-square} and $\beta \geq 0$ we have
\[
\beta \mathcal{A} (W + W^{[3]} + K) (\psi)
\geq \beta ((1-\lambda/2) \norm{\psi}^2_{H^{1 - \varepsilon}} - 1)
    - \frac{1}{8} \norm{\psi}^2_{H^{1 - \varepsilon}}
    - \mathcal{B}^1 (W, \lambda, m, \beta) - \mathcal{E} (K),
\]
where we assume $\norm{\psi}_{L^2} = 1$ and denote for some $\varepsilon' > 0$ small
\begin{align*}
\mathcal B^1(W,\lambda,m,\beta)
    &\coloneqq C\Bigl(\lambda\beta\norm{WW^{[3]}}_{\Holder^{-1/2-\varepsilon}}\Bigr)^3
    +C\Bigl(\beta\lambda^{3/4}
        \norm{W}_{\Holder^{-1/2-\varepsilon}}\Bigr)^{8/(1-\varepsilon')}\\
    &\quad +C\Bigl(\beta\lambda^{7/8}m^{-1/2}
    \norm{W}_{\Holder^{-1/2-\varepsilon}}\Bigr)^{16/(1-\varepsilon')}
    +\beta C_\varepsilon \lambda^{1-2\varepsilon}(1+\lambda)^{r_\varepsilon}\Xi.
\end{align*}
\end{lemma}
\begin{proof}
We have already shown in Lemma~\ref{3d-form-wick-square} that
\begin{equation}
\beta\mathcal A(W)(\psi)
\geq \beta (1-\lambda/2) \norm{\psi}_{H^{1-\varepsilon}}^2
    - \beta C_\varepsilon \lambda^{1-2\varepsilon} (1+\lambda)^{r_\varepsilon}\Xi
    - \beta.
\end{equation}
By definition of $\mathcal A$ we get
\begin{equation}
\begin{split}
\mathcal{A} (W + W^{[3]} + K) (\psi) - \mathcal{A} (W) (\psi)
&=24\lambda\delta^3 \sum_{x \in \Lambda} W W^{[3]} \psi^2
    +24\lambda\delta^3 \sum_{x \in \Lambda} W K\psi^2\\
&\quad+12\lambda\delta^3 \sum_{x\in\Lambda}(W^{[3]}+K)^2\psi^2,
\end{split}
\end{equation}
and since we are deriving a lower bound,
we can ignore the last term as non-negative.
The first term is estimated, after multiplying by $\beta$,
with duality and interpolation as
\begin{equation}
\begin{split}
24\lambda \beta \bigabs{\delta^3 \sum_{x \in \Lambda} W W^{[3]} \psi^2}
&\leq 24 \lambda \beta \norm{W W^{[3]}}_{\Holder^{- 1 / 2 - \varepsilon}}
    \norm{\psi^2}_{B^{1 / 2 + \varepsilon}_{1,1}}\\
&\leq 24 \lambda \beta \norm{W W^{[3]}}_{\Holder^{- 1 / 2 - \varepsilon}}
    \norm{\psi}^{1 + \varepsilon}_{H^{1 - \varepsilon}}\\
&\leq C_\kappa (\lambda \beta \norm{W W^{[3]}}_{\Holder^{- 1 / 2 - \varepsilon}})^3
    + \kappa \norm{\psi}^2_{H^{1 - \varepsilon}}.
\end{split}
\end{equation}

For the second term, the fractional Leibniz rule
(see e.g.~\cite{benyi_fractional_2024}) gives the following estimate,
where the positive parameters $\varepsilon_i$ record arbitrarily
small, interdependent, losses in the Sobolev embeddings:
\begin{equation}\label{eq:estiamte-3d-KW-1}
\begin{split}
&\mathrel{\phantom{=}}
    24 \lambda \beta \bigabs{\delta^3 \sum_{x \in \Lambda} W K \psi^2}\\
&\lesssim \lambda\beta \norm{W}_{\Holder^{-1/2-\varepsilon}}
    \norm{K \psi^2}_{B^{1/2+\varepsilon}_{1,1}}\\
&\lesssim \lambda\beta \norm{W}_{\Holder^{-1/2-\varepsilon}} \Bigl(
    \norm{K}_{L^4} \norm{\psi^2}_{W^{1/2+\varepsilon, 4/3}}
    + \norm{K}_{W^{1/2+\varepsilon, 8/3-\varepsilon_1}}
        \norm{\psi^2}_{L^{8/5+\varepsilon_2}}
\Bigr)\\
&\lesssim \lambda\beta \norm{W}_{\Holder^{-1/2-\varepsilon}} \Bigl(
    \norm{K}_{L^4} \norm{\psi}_{W^{1/2+\varepsilon, 8/3}} \norm{\psi}_{L^{8/3}}
    + \norm{K}_{W^{1/2+\varepsilon, 8/3-\varepsilon_1}}
        \norm{\psi}_{L^{16/5+\varepsilon_3}}^2
\Bigr).
\end{split}
\end{equation}
We now estimate the first term as
\begin{equation}\label{eq:estiamte-3d-KW-2}
\begin{split}
&\mathrel{\phantom{=}} \lambda\beta \norm{W}_{\Holder^{-1/2-\varepsilon}}
    \norm{K}_{L^4} \norm{\psi}_{W^{1/2+\varepsilon, 8/3}} \norm{\psi}_{L^{8/3}}\\
&\lesssim \lambda\beta \norm{W}_{\Holder^{-1/2-\varepsilon}}
    \norm{K}_{L^4} \norm{\psi}_{H^{7/8+\varepsilon}} \norm{\psi}_{H^{3/8}}\\
&\leq \lambda\beta \norm{W}_{\Holder^{-1/2-\varepsilon}}
    \norm{K}_{L^4} \norm{\psi}_{H^{1-\varepsilon}}^{5/4+\varepsilon_4}\\
&\leq C_\kappa (\beta \lambda^{3/4} \norm{W}_{\Holder^{-1/2-\varepsilon}})^{8/(1-4\varepsilon_4)}
    + \kappa \lambda \norm{K}_{L^4}^4
    + \kappa \norm{\psi}_{H^{1-\varepsilon}}^2.
\end{split}
\end{equation}
In the second term we estimate
\begin{equation}\label{eq:estiamte-3d-KW-3}
\begin{split}
&\mathrel{\phantom{=}} \lambda\beta \norm{W}_{\Holder^{-1/2-\varepsilon}}
    \norm{K}_{W^{1/2+\varepsilon, 8/3-\varepsilon_1}}
    \norm{\psi}_{L^{16/5+\varepsilon_3}}^2\\
&\lesssim \lambda\beta \norm{W}_{\Holder^{-1/2-\varepsilon}}
    \norm{K}_{L^4}^{1/2}
    \norm{K}_{H^{1-\varepsilon}}^{1/2}
    \norm{\psi}_{H^{1-\varepsilon}}^{9/8+\varepsilon_5}\\
&= \lambda^{7/8}\beta m^{-1/2} \norm{W}_{\Holder^{-1/2-\varepsilon}}
    \lambda^{1/8} \norm{K}_{L^4}^{1/2}
    m^{1/2} \norm{K}_{H^{1-\varepsilon}}^{1/2}
    \norm{\psi}_{H^{1-\varepsilon}}^{9/8+\varepsilon_5}\\
&\leq C_\kappa (\lambda^{7/8}\beta m^{-1/2} \norm{W}_{\Holder^{-1/2-\varepsilon}})^{16/(1-8\varepsilon_5)}
    + \kappa \lambda \norm{K}_{L^4}^{4}
    + \kappa m^2 \norm{K}_{H^{1-\varepsilon}}^2
    + \kappa \norm{\psi}_{H^{1-\varepsilon}}^2.
\end{split}
\end{equation}
This requires that $\varepsilon$ is small enough compared to $\varepsilon_1$.
Setting $\kappa = 1/24$ gives the stated lower bound.
\end{proof}

\begin{remark}
If we had applied the Boué--Dupuis formula in this setting to deduce exponential tails,
we could conclude with the above lemma as well as the considerations from~\cite{barashkov_variational_2020}.
However this would require to adapt the stochastic estimates in~\cite{barashkov_variational_2020} to the lattice setting,
which is why we opt to use the SPDE-based method where most of the required estimates have already been established.

Since the gradient of $Q(u)$ appears in the method we employ,
we are obliged to deduce some estimates on the ground state before we can proceed.
This will be carried out in the following. 
\end{remark}

\begin{lemma}[Construction of the competitor]\label{lem:competitor}
There exists $\varphi \in L^2 (\Lambda)$ such that
\[
\norm{\varphi}_{\Holder^{1 - \varepsilon}}
    \leq C(1+\norm{\wick{W^2}}_{\Holder^{-1-\varepsilon}}^p),
\qquad \norm{\varphi}_{L^2} = 1,
\]
and $p < \infty$ such that
\[
\abs{\mathcal A(W)(\varphi)}
\leq C\left(1+\norm{W^{2\diamond2}}_{\Holder^{-\varepsilon}}
    + \norm{\wick{W^2}}_{\Holder^{-1-\varepsilon}}\right)^p.
\]
\end{lemma}
\begin{proof}
Let $\varphi$ solve
\begin{equation}\label{eq:competitor fixpoint}
\varphi = -12\lambda (1 - \Laplace_\delta)^{-1} \proj_{>N} (\wick{W^2}\succ \varphi) + 1.
\end{equation}
This equation can be solved by fixed point in~$\Holder^{1 -  \varepsilon}(\Lambda)$
whenever $N \gg \norm{\wick{W^2}}^p_{\Holder^{- 1 -  \varepsilon}}$ for $p$ sufficiently large,
and one gets
\begin{equation}
\norm{\varphi}_{\Holder^{1 - \varepsilon}}
\leq C (1 + \norm{\wick{W^2}}^p_{\Holder^{- 1 - \varepsilon}}).
\end{equation}
Then we have $\norm{\varphi}_{L^2} \geq 1$, since by Jensen's
inequality and the fact that the high-frequency projection has zero mean,
\begin{equation}
\norm{\varphi}_{L^2}
\geq \delta^3 \sum_{x \in \Lambda} \varphi(x)
= 12\lambda \delta^3 \underbrace{\sum_{x \in \Lambda}
        (1 - \Laplace_\delta)^{-1} \proj_{>N} (\wick{W^2}\succ \varphi)(x)}_{=0}
    {} + \delta^3 \sum_{x \in \Lambda} 1.
\end{equation}

Set $h=(1-\Laplace_\delta)^{1/2}\varphi$ and, as in the proof of Lemma~\ref{3d-form-wick-square},
\begin{equation}
w = h+12\lambda(1-\Laplace_\delta)^{-1/2} (\wick{W^2}\succ\varphi).
\end{equation}
Applying $(1-\Laplace_\delta)^{1/2}$ to the fixed-point equation~\eqref{eq:competitor fixpoint} gives
\begin{equation}
w = 1 + 12\lambda(1-\Laplace_\delta)^{-1/2} \proj_{\leq N}(\wick{W^2}\succ\varphi).
\end{equation}
Hence the low-frequency multiplier estimate and the paraproduct estimate
yield
\begin{equation}
\norm{w}_{L^2}
\leq 1+C\lambda N^{2\varepsilon}
    \norm{\wick{W^2}}_{\Holder^{-1-\varepsilon}} \norm{\varphi}_{L^2}
\leq C\left(1+\norm{\wick{W^2}}_{\Holder^{-1-\varepsilon}}\right)^p,
\end{equation}
where we increase~$p$ if necessary.
Finally, the identity~\eqref{eq:anderson-proper} and its remainder estimate give
\begin{equation}
\abs{\mathcal A(W)(\varphi)}
\leq \norm{w}_{L^2}^2 + \abs{R(W,\varphi)} + 1
\leq C\left(1+\norm{W^{2\diamond2}}_{\Holder^{-\varepsilon}}
    + \norm{\wick{W^2}}_{\Holder^{-1-\varepsilon}}\right)^p.
\end{equation}
This upper bound is preserved when we normalize $\varphi$ to have unit $L^2$~norm.
\end{proof}

\begin{proposition}[Bound on ground state]
\label{bound-3d-ground}
Let $\psi^{W + I}$ minimize $\mathcal{A} (W + W^{[3]}+ K) (\psi)$
over $\{\psi \colon \norm{\psi}_{L^2}=1\}$.
Then, for every $\eta \in (0, 1)$ and $\lambda>0$,
\[
\norm{\psi^{W + I}}^2_{H^{1 - \varepsilon}}
\leq \eta \mathcal{E}(K) + \mathcal{B}^2_{\eta} (W, \lambda, m),
\]
where
\[
\begin{split}
\mathcal B^2_{\eta}(W,\lambda,m)
&= \frac{C}{\eta} \Bigl( 1
    + \norm{\wick{W^2}}_{\Holder^{-1-\varepsilon}}
    + \norm{W^{2\diamond2}}_{\Holder^{-\varepsilon}}\\
&\qquad\qquad + \norm{WW^{[3]}}_{\Holder^{-1/2-\varepsilon}}
    + \norm{W}_{\Holder^{-1/2-\varepsilon}}^2
    + \norm{W^{[3]}}_{L^2}^2 \Bigr)^p\\
&\qquad +\mathcal B^1(W,\lambda,m, 16/\eta)
\end{split}
\]
for $p$ sufficiently large.
\end{proposition}

\begin{remark}
The product $W W^{[3]}$ does not possess enough regularity to be well-defined
in a deterministic way in the $\delta \rightarrow 0$ limit.
However one can show, as for example in~\cite[Section~4]{gubinelli_pde_2021}, that
$\sup_{\delta} \E \norm{W W^{[3]}}^p_{\Holder^{- 1 / 2 - \varepsilon}} < \infty$
for any $p < \infty$. 
\end{remark}

\begin{proof}[Proof of Proposition~\ref{bound-3d-ground}]
We begin by deriving an upper bound on $\mathcal{A} (W + W^{[3]} + K) (\psi)$.
If we plug $\varphi$ from Lemma~\ref{lem:competitor} into the functional, we get
\begin{equation}
\begin{split}
\mathcal{A} (W + W^{[3]} + K) (\varphi)
&= \mathcal{A}(W) (\varphi)
    + 24\lambda\delta^3 \sum_{x \in \Lambda} W W^{[3]} \varphi^2\\
&\quad + 24\lambda\delta^3 \sum_{x \in \Lambda} W K \varphi^2
    + 12\lambda\delta^3 \sum_{x \in \Lambda} (W^{[3]} + K)^2 \varphi^2.
\end{split}
\end{equation}
By Lemma~\ref{lem:competitor},
\begin{equation}
\mathcal A(W)(\varphi)
\leq C\left(1+\norm{W^{2\diamond2}}_{\Holder^{-\varepsilon}}
    + \norm{\wick{W^2}}_{\Holder^{-1-\varepsilon}}\right)^p,
\end{equation}
and the same upper bound holds for $\norm{\varphi}_{\Holder^{1-\varepsilon}}$.
The remaining terms satisfy
\begin{equation}
\begin{split}
\bigabs{\langle WW^{[3]},\varphi^2\rangle}
    &\lesssim \norm{WW^{[3]}}_{\Holder^{-1/2-\varepsilon}}
    \norm{\varphi}_{\Holder^{1-\varepsilon}}^2,\\
\abs{\langle WK,\varphi^2\rangle}
    &\lesssim \norm{K}_{H^{1/2+\varepsilon}}
    \norm{W}_{\Holder^{-1/2-\varepsilon}}
    \norm{\varphi}_{\Holder^{1-\varepsilon}}^2,\\
\norm{(W^{[3]}+K)\varphi}_{L^2}^2
    &\lesssim \norm{\varphi}_{\Holder^{1-\varepsilon}}^2
    \left(\norm{W^{[3]}}_{L^2}^2 + \norm{K}_{L^2}^2\right).
\end{split}
\end{equation}
Define
\begin{equation}
\mathcal R_W
= 1 + \norm{\wick{W^2}}_{\Holder^{-1-\varepsilon}}
    + \norm{W^{2\diamond2}}_{\Holder^{-\varepsilon}}.
\end{equation}
Also set
\begin{equation}
\mathcal S_{W,K}
= 1 + \norm{WW^{[3]}}_{\Holder^{-1/2-\varepsilon}}
    + \norm{W^{[3]}}_{L^2}^2
    + \norm{K}_{L^2}^2
    + \norm{K}_{H^{1/2+\varepsilon}} \norm{W}_{\Holder^{-1/2-\varepsilon}}.
\end{equation}
Combining these estimates gives
\begin{equation}
\mathcal A(W+W^{[3]}+K)(\varphi)
\leq \mathcal R_W^{2p} \mathcal S_{W,K}.
\end{equation}
We now use Lemma~\ref{lem:lower-bound}, choosing $\beta_0 = 16/\eta$.
The lemma gives
\begin{equation}
\begin{split}
&\mathrel{\phantom{=}} \mathcal{A} (W + W^{[3]} + K) (\psi^{W + I})\\
&\geq \left(\frac 1 4 \norm{\psi}^2_{H^{1 - \varepsilon}} - 1\right)
- \frac{1}{8\beta_0} \norm{\psi}^2_{H^{1 - \varepsilon}}
- \beta_0^{-1} \mathcal{B}^1 (W, \lambda, m, \beta_0) - \beta_0^{-1}\mathcal{E} (K)\\
&\geq \frac{1}{8}\norm{\psi^{W + I}}^2_{H^{1 - \varepsilon}}
    - \beta_0^{-1}\mathcal{B}^1 (W, \lambda, m,\beta_0)
    - \beta_0^{-1}\mathcal{E} (K)
    - 1.
\end{split}
\end{equation}
Thus we obtain
\begin{equation}
\frac 1 8 \norm{\psi^{W+I}}_{H^{1-\varepsilon}}^2
\leq \mathcal R_W^{2p} \mathcal S_{W,K}
    + \frac{\eta}{16} \mathcal{B}^1 (W, \lambda, m, 16/\eta)
    + \frac{\eta}{16} \mathcal{E} (K)
    + 1.
\end{equation}
Finally, we can use Young's inequality to estimate $\mathcal R_W^{2p} \mathcal S_{W,K}$
with $\tfrac C \eta \mathcal R_W^{4p} + \tfrac{\eta}{16} \mathcal E(K)$.
This finishes the proof.
\end{proof}

We now have all the ingredients to establish Theorem~\ref{thm:3 dimensions}.

\begin{proof}[Proof of Theorem~\ref{thm:3 dimensions}]
Let $\psi^u$ be a normalized ground state of $\mathcal A(u)$. Since
\begin{equation}
Q(u) = -\inf_{\norm{\psi}_{L^2}=1} \mathcal A(u)(\psi),
\end{equation}
the Feynman--Hellmann formula gives, keeping in mind the decomposition $u=W+W^{[3]}+K$, that
\begin{equation}\label{eq:Q-gradient}
\begin{split}
\langle\nabla Q(u),K\rangle
&= -24\lambda\delta^3 \sum_{x\in\Lambda} u(x)K(x)(\psi^u(x))^2.\\
&= -24\lambda\delta^3 \sum_{x\in\Lambda} (W+W^{[3]}+K)(x)K(x)(\psi^u(x))^2\\
&\geq -24\lambda\delta^3 \sum_{x\in\Lambda} (W+W^{[3]})(x)K(x)(\psi^u(x))^2.
\end{split}
\end{equation}
We now apply the same reasoning as in the proof of Lemma~\ref{lem:lower-bound},
more precisely equations~\eqref{eq:estiamte-3d-KW-1},
\eqref{eq:estiamte-3d-KW-2}, and~\eqref{eq:estiamte-3d-KW-3},
as well as those inequalities with $W$ replaced by $W^{[3]}$ to get 
\begin{equation}\label{eq:tilt-pairing}
\begin{split}
&\mathrel{\phantom{=}} q \langle\nabla Q(u),K\rangle  \\ 
&\geq  -C_\eta (\lambda^{7/8}q m^{-1/2} (\norm{W}_{\Holder^{-1/2-\varepsilon}}+\norm{W^{[3]}}_{\Holder^{-1/2-\varepsilon}}))^{16/(1-8\varepsilon_5)}
- \eta \mathcal{E}(K)
- \eta \norm{\psi}_{H^{1-\varepsilon}}^2\\
&\geq -\eta \mathcal{E}(K)- \mathcal{B}^2_{\eta}(W,\lambda,m)
    -C_\eta [\lambda^{7/8}q m^{-1/2} (\norm{W}_{\Holder^{-1/2-\varepsilon}}+\norm{W^{[3]}}_{\Holder^{-1/2-\varepsilon}})]^{16/(1-\varepsilon)},
\end{split}
\end{equation}
where in the last line we have applied Proposition~\ref{bound-3d-ground}.
We set
\begin{equation}
\mathcal{B}^4_{\eta}(W,\lambda,m,q)
\coloneqq C_\eta [\lambda^{7/8}\beta m^{-1/2} (\norm{W}_{\Holder^{-1/2-\varepsilon}}
    +\norm{W^{[3]}}_{\Holder^{-1/2-\varepsilon}})]^{16/(1-8\varepsilon)}.
\end{equation}

We now derive the equation satisfied by $K$, which  follows from the decomposition $K=u-W-W^{[3]}$. Writing
$\mathrm L_\delta=\partial_t+m^2-\Laplace_\delta$ and
$\proj_0 = \mathrm{Id} - \proj_{\geq1}$, we obtain
\begin{equation}
\mathrm L_\delta K
    + 4\lambda\bigl(\wick{u^3}-\wick{W^3}\bigr)
    + 2\lambda^2 a_\delta u
    - q\nabla Q(u)
= \sqrt2\, \proj_0\xi.
\label{eq:K-equation}
\end{equation}
Using Itô's formula we obtain
\begin{equation}
0 = \frac12 \frac{\diff}{\diff t} \E \norm{K}_{L^2}^2
= \E\langle F_K,K\rangle_{L^2} - \operatorname{Tr}_{L^2}(P_0),
\end{equation}
where
\begin{equation}
F_K = (m^2-\Laplace_\delta)K
    + 4\lambda\bigl(\wick{u^3}-\wick{W^3}\bigr)
    + 2\lambda^2a_\delta u-q\nabla Q(u).
\end{equation}
Since $\proj_0$ projects onto the normalized constant mode,
$\operatorname{Tr}_{L^2}(P_0)=1$.
Proposition~\ref{prop-gub-hov} and~\eqref{eq:tilt-pairing} therefore give
\begin{equation}
(1-\eta)\E\mathcal E(K)
\leq 1+\E\!\left[ \mathcal B^3(W,\lambda,m,q)
    + \mathcal B^2_{\eta}(W,\lambda,m)
    + \mathcal B^4_{\eta}(W,\lambda,m)\right].
\label{inequality-energy}
\end{equation}
Finally observe that if $\norm{\psi}_{L^2}=1$,
we have $\norm{\psi}^2_{H^{1-\varepsilon}}\geq 1$, so
\begin{equation}
q \left((1-\lambda/2) \norm{\psi}^2_{H^{1-\varepsilon}}-1\right)
\geq -\frac{\lambda q}{2}.
\end{equation}
Thus we obtain, applying Lemma~\ref{lem:lower-bound},
\begin{equation}
\begin{split}
q\int Q(u)\,\diff\nu^{qQ}(u)
&= -q\,\E\inf_{\|\psi\|_2=1}\mathcal A
       (W+W^{[3]}+K)(\psi)\\
&\leq C\, \E\!\left[1 +\lambda q/2 +\mathcal E(K)
    + \mathcal B^1(W,\lambda,m,q)\right],
\end{split}
\end{equation}
where the inequality is the negative of the lower bound in Lemma~\ref{lem:lower-bound}.
Using the energy estimate~\eqref{inequality-energy} and the asserted uniform moments of~$\mathcal B^i$,
this is the required \emph{upper} bound on $q\int Q\,d\nu^{qQ}$.
It can then be fed into Lemma~\ref{lem:hairer-steele} to bound $Z^{qQ}$, and obtain $\mathcal{R}$.

Finally by examining each $\mathcal B^i$, we observe that $\mathcal R$ indeed satisfies
\begin{equation}
\sup_{q, m, \lambda}  \mathcal{R}(q,m,\lambda) < \infty
\end{equation}
in the regime $\lambda\lesssim m^{3}$, $m\leq 1$, $q\lesssim m^{-2}$.
\end{proof}

\appendix
\section{Besov spaces on the lattice}\label{sec:besov}

We summarize here our conventions and some basic results on
Fourier analysis and Besov spaces on lattices.
For the most part we follow~\cite{martin_paracontrolled_2019,mourrat_convergence_2017}
but use somewhat different notation.

Given the periodic lattice $\Lambda \coloneqq \delta \Z^d \cap [0, L)^d$,
the Fourier transform is defined on the reciprocal lattice
\begin{equation}
\hat\Lambda \coloneqq \frac{2\pi}{L} \Z^d \cap \Bigl[-\frac\pi\delta, \frac\pi\delta \Bigr)^d
\end{equation}
via the formula
\begin{equation}
\hat f(k) = \left( \frac \delta L \right)^{\! d} \sum_{x \in \Lambda}
    \exp(-i k \cdot x) f(x),
\quad k \in \hat\Lambda.
\end{equation}
Conversely the inverse Fourier transform is defined as
\begin{equation}
f(x) = \sum_{k \in \hat\Lambda}
    \exp(i k \cdot x) \hat f(k),
\quad x \in \Lambda.
\end{equation}
In this convention Parseval's identity reads
\begin{equation}\label{eq:parseval}
\delta^d \sum_{x \in \Lambda} \abs{f(x)}^2
= L^d \sum_{k \in \hat\Lambda} \abs{\hat f(k)}^2.
\end{equation}

Let us then define Besov spaces.
In order to get uniform-in-volume bounds in Theorem~\ref{thm:integrability},
we equip subsets of $\R^2$ with a weight.
Even though many classes of weights would be possible,
we only consider polynomial weights of form
\begin{equation}
\rho(x) = (1 + \abs x^2)^{-\alpha/2}
\end{equation}
for $\alpha \geq 0$.
In particular $\alpha = 0$ gives a flat weight.

\begin{definition}[Besov spaces]
Let $\{ \Delta_j \}_{j = -1, \ldots, j_\Lambda}$ be the Littlewood--Paley blocks
as in~\cite[Definition~2.17]{martin_paracontrolled_2019}.
The weighted Besov space $B^s_{p,r}(\rho)$ on $\Lambda$ is defined
as the completion of~$C^\infty_c(\Lambda)$ with respect to the norm
\[
\norm{f}_{B^s_{p,r}(\rho)}
\coloneqq \bignorm{2^{js} \norm{\rho \Delta_j f}_{L^p(\Lambda)}}_{\ell^r(j)}.
\]
On $\Tor^d$ and $\R^d$ the definition is analogous.
We abbreviate $H^s \coloneqq B^s_{2,2}$ and $\Holder^s \coloneqq B^s_{\infty, \infty}$.
\end{definition}

\begin{theorem}\label{thm:besov properties}
Let the ambient space be either~$\Lambda_\delta$ or~$\Tor^2$,
let $\rho_1$ and $\rho_2$ be weights,
let $p = 1/p_1 + 1/p_2$ and $r = 1/r_1 + 1/r_2$ be Hölder conjugate pairs in~$[1,\infty]$,
and let $s \in \R$.
Then the following hold:
\begin{enumerate}
\renewcommand{\labelenumi}{(\alph{enumi})}
\item Multiplication: $\norm{fg}_{B^s_{p,r}(\rho_1 \rho_2)}
\lesssim \norm{f}_{B^s_{p_1, r}(\rho_1)} \norm{g}_{B^t_{p_2, r}(\rho_2)}$
    for $t \geq s$ such that $s + t > 0$.
\item Duality: $\norm{\rho_1 \rho_2 fg}_{L^1}
        \lesssim \norm{f}_{B^{-s}_{p_1, r_1}(\rho_1)} \norm{g}_{B^s_{p_2, r_2}(\rho_2)}$,
    where we set $p = r = 1$.
\item Interpolation: let $\theta \in (0,1)$ and
    \[
        s = \theta s_1 + (1-\theta) s_2, \quad
        \frac 1 p = \frac{\theta}{p_1} + \frac{1-\theta}{p_2}, \quad
        \frac 1 r = \frac{\theta}{r_1} + \frac{1-\theta}{r_2}.
    \]
    Then $\norm{f}_{B^s_{p,r}(\rho_1^{\theta} \rho_2^{1-\theta})}
        \leq \norm{f}_{B^{s_1}_{p_1, r_1}(\rho_1)}^\theta
            \norm{f}_{B^{s_2}_{p_2, r_2}(\rho_2)}^{1-\theta}$.
\item Besov embedding: $\norm{f}_{B^s_{p,r}(\rho_1)} \lesssim \norm{f}_{B^t_{q,r}(\rho_1)}$
    where $q < p$ and $t \geq s + d(1/q - 1/p)$.
\item $L^p$ embedding: $\norm{f}_{B^0_{p,\infty}(\rho_1)} \lesssim \norm{f}_{L^p(\rho_1)}$.
\item Second parameter tradeoff:
    $\norm{f}_{B^s_{p,r}(\rho_1)} \lesssim \norm{f}_{B^{s+\varepsilon}_{p,\infty}(\rho_1)}$
    for all $1 \leq r \leq \infty$ and $\varepsilon > 0$.
\end{enumerate}
The implicit constants are uniform in $\delta$ (small enough).
\end{theorem}
\begin{proof}
The classical reference for these results is~\cite[Sections~2.7--2.9]{bahouri_fourier_2011}.
The lattice versions follow from discrete counterparts of the necessary tools
(like Young's convolution inequality).
\end{proof}

The Besov norm resembles the fractional Sobolev norm
\begin{equation}
\norm{f}_{W^{s,p}(\rho)} \coloneqq \norm{\rho (1 - \Laplace)^{s/2} f}_{L^p},
\quad s \in \R,\; 1 \leq p \leq \infty.
\end{equation}
The two norms are equivalent when $p = r = 2$ and the space is equipped with a flat weight:

\begin{lemma}[Equivalent $H^s$ norm]\label{thm:besov hs norm}
For all $\delta \geq 0$ and $f \in H^s(\Lambda_\delta)$ we have
\[
\norm{f}_{H^s(\Lambda_\delta)}^2
\simeq \norm{(1 - \Laplace_\delta)^{s/2} f}_{L^2(\Lambda_\delta)}^2
\simeq \sum_{k \in \hat\Lambda_\delta} (1 + \abs k^2)^s \abs{\hat f(k)}^2.
\]
\end{lemma}
\begin{proof}
The Fourier symbols of $\Delta_j$ are bounded and only finitely many of them overlap at any point.
Both equivalences then follow from Plancherel's identity.
\end{proof}

We also use the following results on the paraproduct decomposition.
For the definitions, we refer to~\cite{bahouri_fourier_2011,mourrat_global_2017}.

\begin{theorem}[Paraproducts]\label{thm:besov paraproducts}
The product of two smooth functions can be decomposed into two paraproducts and resonant product
\[
fg = f \prec g + f \succ g + f \odot g
\]
as in~\cite[Section~2.8]{bahouri_fourier_2011}.
We also denote $f \preceq g \coloneqq f \prec g + f \odot g$.

For any weights $\rho_1$, $\rho_2$,
and $r$, $p$, $p_1$, $p_2 \in [1, \infty]$ such that $1/p = 1/p_1 + 1/p_2$,
the paraproduct satisfies the estimate
\[
\norm{f \prec g}_{B^{\min(\alpha,0) + \beta}_{p,r}(\rho_1 \rho_2)}
\lesssim \norm{f}_{B^\alpha_{p_1, \infty}(\rho_1)}
    \norm{g}_{B^\beta_{p_2, r}(\rho_2)},
\quad \alpha, \beta \in \R,
\]
and the resonant product satisfies
\[
\norm{f \odot g}_{B^{\alpha + \beta}_{p,r}(\rho_1 \rho_2)}
\lesssim \norm{f}_{B^\alpha_{p_1, \infty}(\rho_1)}
    \norm{g}_{B^\beta_{p_2, r}(\rho_2)},
\quad \alpha + \beta > 0.
\]
\end{theorem}
\begin{proof}
\cite[Theorem~3.1]{mourrat_global_2017},
again adapted to the lattice.
\end{proof}

\section{Wick ordering and stochastic estimates}\label{sec:wick}

In this appendix we collect the basic definitions and properties related to Wick renormalization.
For motivation and further details we refer to~\cite[Chapter~1]{simon_euclidean_1974}.

\begin{definition}[Wick ordering]\label{def:wick}
Let $W$ be a Gaussian random variable with mean zero, and let $a = \E W^2$.
We define the first four Wick powers as
\[
\wick{W} = W, \quad
\wick{W^2} = W^2 - a, \quad
\wick{W^3} = W^3 - 3 a W, \quad
\wick{W^4} = W^4 - 6 a W^2 + 3 a^2,
\]
and for any $\beta > 0$ we define
\[
\wick{\sin(\beta W)} = \exp\!\left( \frac{\beta^2 a}{2} \right) \sin(\beta W), \quad
\wick{\cos(\beta W)} = \exp\!\left( \frac{\beta^2 a}{2} \right) \cos(\beta W).
\]
\end{definition}

\begin{theorem}[Wick formula]
Let $X$ and $Y$ be centred Gaussian random variables.
Then
\[
\E [\wick{X^n} \wick{Y^m}] = n! \I_{n=m} (\E [XY])^n
\]
for all $n, m \in \N$.
In particular $\E \wick{X^n} = 0$ for all $n > 0$.
\end{theorem}

The Gaussian free field is almost surely distribution-valued in two or more dimensions.
The Wick powers and sines/cosines of the GFF are however well-defined
as limits of smooth approximations.
See for example~\cite[Section~5]{mourrat_global_2017}.

For brevity, we use the integral notation for convolutions
even though $\dx$ and $\dy$ refer to the normalized point measures on~$\Lambda_\delta$.

\begin{theorem}[Moments of two-dimensional GFF]\label{wick-power-estimate}
Let $\Lambda_\delta \subset \R^2$,
and let $W$ be distributed according to the Gaussian free field on~$\Lambda_\delta$
with mass~$m^2$.
Then for $k \in \N$ and $p < \infty$ we have
\[
\E \norm{\wick{W^k}}_{\Holder^{-\kappa}(\rho)}^p \lesssim 1 + \abs{\log(m)}^{pk/2}.
\]
Here we assume that $\rho$ is sufficiently integrable over $\Lambda_\delta$.
\end{theorem}
\begin{proof}
We first use Theorem~\ref{thm:besov properties}
to go from $\Holder^{-\kappa}$ to $B_{p,p}^{-\kappa/2}$ for $p$ large enough.
We then use hypercontractivity, translation invariance, and Wick formula
to estimate
\begin{equation}
\begin{split}
\E \norm{\wick{W^k}}_{B_{p,p}^{-\kappa}(\rho)}^p
&\lesssim \sum_{j \geq -1} 2^{-pj\kappa} [\E (\Delta_j \wick{W^k}(0))^2]^{p/2} \int \rho(x)^p \dx\\
&\lesssim \sum_{j \geq -1} 2^{-pj\kappa} \int K_j(-x) K_j(-y) \E [\wick{W^k}(x) \, \wick{W^k}(y)]^{p/2} \dx \dy\\
&\lesssim \sum_{j \geq -1} 2^{-pj\kappa} \int K_j(-x) K_j(-y) G(x,y)^{pk/2} \dx \dy,
\end{split}
\end{equation}
where $K_j$ are the kernels of the Littlewood--Paley blocks $\Delta_j$,
and $G$ is the Green's function of the free field.

The Green's function is of order $\log(m \abs{x-y})$ near the diagonal
and rapidly decaying away from it~\cite[Chapter~7]{glimm_quantum_1981}.
On the other hand, the kernels $K_j$ are uniformly bounded in $L^1$ norm
and of order $2^j$ in $L^\infty$ norm.
Therefore we can interpolate
\begin{equation}
\norm{K_j}_{L^{q/(q-1)}} \lesssim 2^{j/q}.
\end{equation}
These considerations imply that
\begin{equation}
\E \norm{\wick{W^k}}_{B_{p,p}^{-\kappa}(\rho)}^p
\lesssim \sum_{j \geq -1} 2^{-pj\kappa} (\abs{\log(m)}^{pk/2} + 2^{j/q}).
\end{equation}
By choosing $q$ large enough, we get the claim.
\end{proof}

\begin{theorem}[Moments of two-dimensional GFF sine]\label{moment-gff-sine}
Let $W$ and~$\rho$ be as above, and $\beta^2 < 4\pi$.
Then
\[
\E \norm{\wick{\sin(\beta W)}}_{\Holder^{-\beta^2 / 4\pi -\kappa}(\rho)}^p
\lesssim (m\wedge1)^{-p \beta^2 / 4\pi}
\]
for sufficiently large $1 < p < \infty$. The same bound holds for the cosine. 
\end{theorem}
\begin{proof}
We express $W = W_{\infty} = \int^{\infty}_0 J_t \diff B_t$,
where $J_t$ is as in Section~\ref{sec:sine-gordon measure}
and $B_t$ is an $L^2 (\Lambda)$ cylindrical Brownian motion.
It follows from the Itô formula that
$\wick{\sin (\beta W_\infty)} = \beta \int_0^\infty \wick{\cos (\beta W_t)} \diff W_t$.
Its quadratic variation is given by
\begin{equation}
\begin{split}
&\mathrel{\phantom{=}}
\beta^2 \left\langle \int^{\infty}_0 \wick{\cos (\beta W_t)} (x) \diff W_t (x),
    \int^{\infty}_0 \wick{\cos (\beta W_t)} (y) \diff W_t (y)
\right\rangle\\
& = \beta^2 \int_0^\infty C^{\ast 2}_t (x - y) \wick{\cos (\beta W_t) (x)}
    \wick{\cos (\beta W_t) (y)} \diff t.
\end{split}
\end{equation}
Here $C^{\ast 2}_t$ is the convolution kernel of $J_t^2$,
as in Section~\ref{sec:sine-gordon measure}.

Using this we can deduce that the Littlewood--Paley blocks satisfy
\begin{equation}
\begin{split}
&\mathrel{\phantom{=}} \langle (\Delta_j \wick{\sin (\beta W)}) (z) \rangle\\
&= \beta^2 \int^{\infty}_0 \!\!\! \iint C^{\ast 2}_t(x - y) K_j (x - z) K_j (y - z)
    \wick{\cos(\beta W_t) (x)} \wick{\cos (\beta W_t) (y)} \dx \dy \diff t\\
&\leq \beta^2 \int^{\infty}_0 \!\!\! \iint C^{\ast 2}_t(x - y)
    \abs{K_j (x- z) K_j (y - z)} \inorm{t}^{\beta^2 / 2 \pi} \dx \dy \diff t\\
&= \beta^2 \int^{\infty}_0 \!\!\! \iint C^{\ast 2}_t(x - y)
    \abs{K_j (x) K_j (y)} \inorm{t}^{\beta^2 / 2 \pi} \dx \dy \diff t,
\end{split}
\end{equation}
where $K_j$ is again the convolution kernel of the block $\Delta_j$.

Now it is not hard to show that (see~\cite[Appendix~A]{bauerschmidt_logsobolev_2021})
\begin{equation}
\begin{split}
\int_0^\infty C^{\ast 2}_t(x - y) \inorm{t}^{\beta^2 / 2 \pi} \diff t
&\lesssim \int_0^\infty e^{- \frac{\abs{x - y}^2}{2 t} - m t} \inorm{t}^{\beta^2 / 2 \pi} \diff t\\
&\lesssim \frac{1}{\abs{x-y}^{\beta^2 / 2 \pi}} + m^{\beta^2 / 2 \pi}.
\end{split}
\end{equation}
Now for $2 / q > \beta^2 / 4 \pi$ we can estimate
\begin{equation}
\begin{split}
&\mathrel{\phantom{=}} \iint
    \left( \frac{1}{| x - y |^{\beta^2 / 2 \pi}} + m^{\beta^2 / 2\pi} \right)
    \abs{K_j (x) K_j (y)} \inorm{t}^{\beta^2 / 2 \pi} \dx \dy\\
&\lesssim \norm{K_j}^2_{L^q} + m^{-\beta^2 / 2 \pi} \norm{K_j}^2_{L^1}.
\end{split}
\end{equation}
From this we can conclude by applying the Burkholder--Davis--Gundy inequality:
for $\kappa > \beta^2 / 4 \pi$ and $p$ large enough, we have
\begin{equation}
\begin{split}
&\mathrel{\phantom{=}} \E \norm{\wick{\sin (\beta W)}}^p_{\Holder^{-\kappa}(\rho)}\\
&\lesssim \E \norm{\wick{\sin (\beta W)}}^p_{B_{p, p}^{- \kappa + \varepsilon / 2}(\rho)}\\
&\leq \sum 2^{- j p (\kappa + \varepsilon / 2)}
    \E \norm{\Delta_j \wick{\sin (\beta W)}}_{L^p (\rho)}^p\\
&\leq \sum 2^{- j p (\kappa + \varepsilon / 2)}
    \E \abs{\Delta_j \wick{\sin (\beta W)} (0) }^p\\
&\lesssim \sum 2^{- j p (\kappa + \varepsilon / 2)}
    (\E \langle \Delta_j \wick{\sin (\beta W)} (0) \rangle)^{p / 2}\\
&\lesssim \sum 2^{- j p (\kappa + \varepsilon / 2)}
    (\norm{K_j}^p_{L^q} + m^{-p\beta^2 / 4 \pi} \norm{K_j}^p_{L^1})\\
&\lesssim (m \wedge 1)^{-p \beta^2 / 4 \pi}.
\end{split}
\end{equation}
The proof for the $\cos$ is analogous. 
\end{proof}

\begin{lemma}[Wick binomial rule]\label{thm:wick binomial}
Let $W$ be sampled from the Gaussian free field $\nu$,
and assume that $v \in L^{2q}(\nu;\, B^{2\varepsilon}_{pk,p}(\rho))$ for some $p, q \in [1, \infty)$
and integrable weight $\rho$.
Then
\[
\wick{(W + v)^k} = \sum_{j=0}^k \binom k j \wick{W^j} v^{k-j}
\]
as an element of $L^q(\nu, B^{-\varepsilon}_{p,p}(\rho^k))$.
\end{lemma}
\begin{proof}
We only use the cases $k \leq 3$ that are computable by hand.
\end{proof}

\bibliographystyle{plainnat}
\bibliography{SdNLW-ref}

\end{document}